\documentclass[12pt,draftcls,onecolumn]{IEEEtran}
\usepackage{amsmath}
\usepackage{amsthm}
\usepackage{amsfonts}
\usepackage{amssymb}
\usepackage[utf8]{inputenc}
\DeclareUnicodeCharacter{00A0}{ }
\usepackage{graphicx}
\usepackage{graphics}
\usepackage{float}
\usepackage{tikz}
\usepackage{svg}
\usetikzlibrary{shapes,snakes,patterns}
\usepackage{epstopdf}

\usepackage[justification=centering]{caption}
\usepackage{enumerate} 
\usepackage{cite}
\usepackage{placeins}
\newcommand{\vast}{\bBigg@{4}}
\newcommand{\Vast}{\bBigg@{5}}
\newtheorem{theorem}{Theorem}[]
\newtheorem{lemma}[]{Lemma}
\newtheorem{corollary}{Corollary}[theorem]
\usepackage{suffix}
\usepackage{mathtools}
\usepackage{amsthm}
\usepackage{enumitem}
\usepackage{xfrac}
\usepackage{tikz}
\usepackage{relsize}
\theoremstyle{definition}

\usepackage[list=true]{subcaption}
\usepackage[xcolor,pdftex]{changebar}
\usepackage[normalem]{ulem}
\begin{document}
\bstctlcite{IEEEexample:BSTcontrol}
\title{Multi-RIS Communication Systems: Asymptotic analysis of best RIS selection for i.n.i.d. Random Variables using Extreme Value Theory  }
	\author{ Srinivas Sagar and  Sheetal Kalyani \\
      \thanks{\hspace{-0.7cm} \\
The authors are with the Department of Electrical Engineering, Indian Institute of Technology Madras, India. Emails: \{ee21d051@smail, 
skalyani@ee\}.iitm.ac.in}}

	\maketitle
	\begin{abstract}
		\textcolor{black}{This paper investigates the performance of multiple reconfigurable intelligent surfaces (multi-RIS) communication systems where the RIS link with the highest signal-to-noise-
        ratio (SNR) is selected at the destination. In practice, all the RISs will not have the same number of reflecting elements. Hence, selecting the RIS link with the highest SNR will involve characterizing the distribution of the maximum of independent, non-identically distributed (i.n.i.d.) SNR random variables (RVs).  Using extreme value theory (EVT), we derive the asymptotic distribution of the normalized maximum of i.n.i.d. non-central chi-square (NCCS) distributed SNR RVs with one degree of freedom (d.o.f) and then extend the results for $k$-th order statistics. Using these asymptotic results, the outage capacity and average throughput expressions are derived for the multi-RIS system.} The results for independent and identically
        distributed (i.i.d.) SNR RVs are then derived as a special case of i.n.i.d. RVs. All the derivations are validated through extensive Monte Carlo simulations, and their utility is discussed. 
	\end{abstract}
 
	
	\section{Introduction}
 Reconfigurable intelligent surfaces (RISs) have gained significant popularity in the last few years \cite{wu2019intelligent,wu2019beamforming,abeywickrama2020intelligent,pan2020multicell,pan2020intelligent,zhang2020capacity,zhou2020robust}. Various applications such as beamforming \cite{wu2019intelligent, di2020hybrid, wu2020joint}, 
 multiple input multiple output (MIMO) \cite{pan2020multicell,pan2020intelligent}, deep learning \cite{huang2020reconfigurable},
 unmanned aerial vehicles (UAVs) \cite{li2020reconfigurable}, simultaneous wireless information and power transfer (SWIPT) system \cite{pan2020intelligent}, non-orthogonal multiple access (NOMA) \cite{zheng2020intelligent}, and millimeter wave (mmWave) systems \cite{{gopi2020intelligent}} now use the RIS. 
 \par RIS-aided communication can be grouped into two categories: single-RIS and multi-RIS systems. Performance analysis of single-RIS communication systems is extensively studied in \cite{wu2019intelligent,di2020hybrid,wu2020joint,pan2020multicell,zhang2020capacity,zhou2020robust,huang2020reconfigurable, taha2021enabling,yang2020deep,charishma2021outage, subhash2023max, subhash2023optimal,shekhar2022instantaneous}. To improve the system performance works like \cite{wu2019intelligent,di2020hybrid,wu2020joint} presented single RIS communication systems in multi-antenna transmitter \cite{wu2019intelligent}, multiuser communication  \cite{di2020hybrid},  and SWIPT  \cite{wu2020joint} to minimize the transmit power,  hybrid beamforming, active and passive beamforming, respectively.  

 The results of \cite{basar2019wireless} show that, in an  RIS-aided communications system, received SNR at the destination can be modeled as non-central chi-square (NCCS) distribution with one degree of freedom (d.o.f) where the parameters $\left (\lambda =\left ( \frac{N\pi }{4} \right )^{2}, \sigma ^{2}=N\left ( 1-\frac{\pi ^{2}}{16} \right )\right )$ only depend on the number of reflecting elements $(N)$ of RIS. The authors of \cite{yang2020coverage} introduced the quantitative analysis of coverage area. Here, the source communicates to the destination through a single RIS. 

\par In a multi-RIS system, multiple RIS links are available between the source and destination along with a direct link. Works like \cite{do2021multi,phan2022performance,xie2022downlink,tran2022exploiting,tran2022combining,agbogla2023adaptive,nguyen2023secrecy,nguyen2023performance} used the direct link and reflected links from all the RISs in deriving the outage probability and throughput analysis of multi-RIS systems. Authors of \cite{do2021multi} considered the statistical characterization of exhaustive RIS-aided (transmit the RIS signal along with a direct signal) and opportunistic RIS-aided (only the best RIS along with a direct signal) systems. Multi-RIS systems are explored in several communications applications like cooperative communication  \cite{phan2022performance}, single cell networks with backhaul capacity \cite{xie2022downlink}, cooperative RIS, and opportunistic RIS methods for terahertz communication systems \cite{agbogla2023adaptive}. Also, the performance analysis of a multi-RIS system was presented in \cite{tran2022exploiting, tran2022combining,nguyen2023secrecy }. The authors of \cite{nguyen2023performance} derived outage probability expression for the UAV-NOMA-mmWave system with multiple RISs. Multi-RIS systems are also useful when no direct link is available between source and destination. Several works \cite{fang2022optimum,aldababsa2023multiple,hindustani2023outage,tam2023improving} select the RIS link, giving the highest SNR at the destination. The authors derived outage probability and throughput expressions using opportunistic RIS selection for SISO \cite{hindustani2023outage} and mmWave \cite{tam2023improving} systems.
\par The multi-RIS system considered in \cite{yang2020outage} assumed that all the RISs would have the same number of reflecting elements leading to i.i.d. SNR links. Then, the authors used EVT to characterize the order statistics of the received SNR distribution. But in practice, the number of reflecting elements in each RIS will be different, leading to i.n.i.d. SNR RVs. The primary motivation of our paper is to characterize the order statistics of i.n.i.d  NCCS RVs with one d.o.f and use that to select the best RIS among the multiple RISs.
\par
The primary focus of EVT is the statistical characterization of extreme (maximum/minimum) values. EVT results have been extensively used in the fields of communication like multiuser diversity \cite{park2008performance,seo2009new}, cognitive radio (CR) \cite{xia2013spectrum,haider2015spectral}, relays \cite{xia2013spectrum}, MIMO, ultra-reliable and low-latency communication (URLLC) \cite{liu2019dynamic, samarakoon2019distributed}, and machine learning. 
In most scenarios, finding the exact CDF of the maximum order statistics, i.e., $\prod_{r=1}^{R}F_{r}\left ( \gamma  \right )$ leads to very complicated expressions for large $R$. 
Several authors used EVT to characterize the maximum order statistic since this leads to mathematically tractable expressions.  
In multiuser diversity systems, throughput analysis is carried out asymptotically with the help of EVT \cite{park2008performance,seo2009new}. CR systems used EVT to find the limiting distribution of end-to-end SNR \cite{xia2013spectrum},  to analyze spectrum and energy efficiency \cite{haider2015spectral}, and for optimum power allocation \cite{subhash2020transmit}.  Also, in URLLC systems, EVT has been used to characterize the tail distribution of queue length \cite{liu2019dynamic, samarakoon2019distributed}. The works in \cite{kalyani2012asymptotic, subhash2019asymptotic} presented the asymptotic distribution of maximum order statistics for i.i.d. sums of non-identical gamma RVs \cite{kalyani2012asymptotic} and $\kappa-\mu$ shadow fading RVs \cite{subhash2019asymptotic} respectively.
\par To the best of our knowledge, while EVT has been used extensively in characterizing communication systems, the focus has been on i.i.d. RVs. Characterization of the asymptotic distribution of order statistics of i.n.i.d. RVs \cite{mejzler1969some,barakat2013limit} is mathematically more complicated than the characterization of i.i.d. RVs. Our work in \cite{subhash2021cooperative} was the first to consider i.n.i.d. RVs in the context of an opportunistic relaying system with the SWIPT network.
We then further derived the order statistics of i.n.i.d. Rician fading RVs \cite{subhash2022asymptotic}. From \cite{basar2019wireless}, we can observe that end-to-end SNR in RIS-aided communication system can be modeled as NCCS RV with one d.o.f, so we would like to characterize the order statistics of NCCS RV with one d.o.f. Given the order statistics of RVs, many applications (selection diversity, relay selection, antenna selection) in communications select the maximum order statistic for communication. Sometimes, the best selection/maximum order statistic may not be available for communication, in such scenarios, $k$-th best link can be selected for communication.
Hence, we would also like to study the order statistics of NCCS RV with one d.o.f to characterize the end-to-end SNR of a multi-RIS system. Now, we present the main contributions of this paper.
\begin{enumerate}
    \item We present the performance analysis of a multi-RIS system with the help of EVT. Modeling the end-to-end SNR of a RIS-aided communication system as NCCS RV with one d.o.f, we first derive the asymptotic distribution of the normalized maximum of $R$ i.n.i.d. NCCS RVs with the help of EVT. The asymptotic distribution of $k$-th maximum of i.n.i.d. NCCS RVs is also derived.
    \item Assumption of the same number of reflecting elements for all the RISs gives us i.i.d. NCCS RVs. So, the asymptotic distribution of $k$-th order statistics for $R$ i.i.d. NCCS RVs with one d.o.f is also derived as a special case of i.n.i.d. RVs.
    \item Using the asymptotic distribution of maximum order statistics of SNR RV, average throughput and outage capacity expressions are derived for RIS-aided communication systems considering multiple RISs. Stochastic ordering results for the normalized $k$-th maximum SNR RV are also presented.
\end{enumerate}
\par
The organization of the paper is as follows. Section II introduces the considered system model for RIS-aided communication systems. Section III presents the results of $k$-th maximum order statistics of NCCS RVs with one d.o.f and derives the average throughput and outage capacity expressions for RIS-aided communication systems in a multi-RIS scenario. Section IV provides extensive simulation results to support our theoretical analysis, and Section V concludes the paper.
\par
The following notations are used in the paper. The probability density function and cumulative distribution functions of an RV $X$ are denoted by $f_{X}\left ( . \right )$ and $F_{X}\left ( . \right )$, respectively. The expectation of RV $X$ is denoted as $\mathbb{E}\left ( X \right )$. Given an event $A$, $\mathbb{P}\left ( A \right )$ denotes the probability of the event $A$.
    \section{System Model}
    In this work, we consider a RIS-aided wireless communication system model as shown in Fig. \ref{fig:sys}. It has a source (S), destination (D), and R number of RISs. The source and destination have a single antenna, and the $r^{th}$ RIS has $N_r$ reflecting elements. Different RISs can have different numbers of reflecting elements. Similar to \cite{yang2020outage}, we assume that there is no direct link between the source and destination due to the outage. Here, RISs act as passive reflectors between the source and destination and improve the quality of the signal at the receiver. Each RIS will reflect the signal transmitted by the source to the destination, so there are $R$ links available at the receiver for processing.  Similar to \cite{do2021multi,yang2020outage,yildirim2020modeling,galappaththige2020performance}, we assume that each RIS is controlled to steer their beam to the destination, avoiding interference with each other. The link with the highest quality is selected for communication between the source and destination in an opportunistic multi-RIS environment \cite{yang2020outage}.

    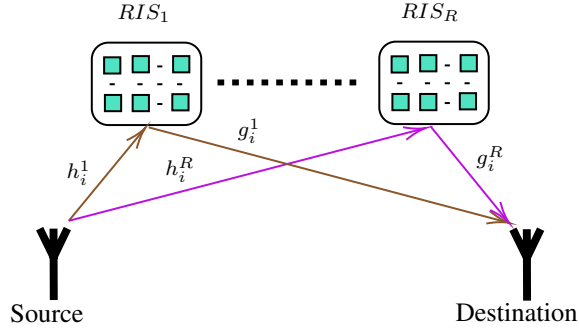
\begin{figure}[h]
        \centering        
        \tikzset{every picture/.style={line width=0.75pt}} 
        
        \begin{tikzpicture}[x=0.75pt,y=0.75pt,yscale=-1,xscale=1]
        
        \draw  [line width=3]  (143.3,153.8) -- (150.34,168.19) -- (157.28,153.75) (150.29,153.77) -- (150.42,189.82) ;
        \draw  [line width=3]  (382,154.1) -- (389.04,168.49) -- (395.98,154.05) (388.99,154.07) -- (389.12,190.12) ;
        \draw   (169.6,68.22) .. controls (169.6,63.62) and (173.32,59.9) .. (177.92,59.9) -- (216.48,59.9) .. controls (221.08,59.9) and (224.8,63.62) .. (224.8,68.22) -- (224.8,93.18) .. controls (224.8,97.78) and (221.08,101.5) .. (216.48,101.5) -- (177.92,101.5) .. controls (173.32,101.5) and (169.6,97.78) .. (169.6,93.18) -- cycle ;
        \draw  [fill={rgb, 255:red, 80; green, 227; blue, 194 }  ,fill opacity=1 ] (175.8,67.3) -- (184,67.3) -- (184,75.5) -- (175.8,75.5) -- cycle ;
        \draw  [fill={rgb, 255:red, 80; green, 227; blue, 194 }  ,fill opacity=1 ] (190.2,66.7) -- (198.4,66.7) -- (198.4,74.9) -- (190.2,74.9) -- cycle ;
        \draw  [fill={rgb, 255:red, 80; green, 227; blue, 194 }  ,fill opacity=1 ] (210.6,66.7) -- (218.8,66.7) -- (218.8,74.9) -- (210.6,74.9) -- cycle ;
        \draw  [fill={rgb, 255:red, 80; green, 227; blue, 194 }  ,fill opacity=1 ] (176.2,86.3) -- (184.4,86.3) -- (184.4,94.5) -- (176.2,94.5) -- cycle ;
        \draw  [fill={rgb, 255:red, 80; green, 227; blue, 194 }  ,fill opacity=1 ] (190.2,86.3) -- (198.4,86.3) -- (198.4,94.5) -- (190.2,94.5) -- cycle ;
        \draw  [fill={rgb, 255:red, 80; green, 227; blue, 194 }  ,fill opacity=1 ] (210.2,85.9) -- (218.4,85.9) -- (218.4,94.1) -- (210.2,94.1) -- cycle ;
        \draw    (205.2,71.9) -- (202.6,71.9) ;
        \draw    (205.2,90.9) -- (202.6,90.9) ;
        \draw    (181.2,80.9) -- (178.6,80.9) ;
        \draw    (195.2,80.9) -- (192.6,80.9) ;
        \draw    (215.2,80.9) -- (212.6,80.9) ;
        \draw    (205.2,80.9) -- (202.6,80.9) ;
        
        \draw   (314.6,67.72) .. controls (314.6,63.12) and (318.32,59.4) .. (322.92,59.4) -- (361.48,59.4) .. controls (366.08,59.4) and (369.8,63.12) .. (369.8,67.72) -- (369.8,92.68) .. controls (369.8,97.28) and (366.08,101) .. (361.48,101) -- (322.92,101) .. controls (318.32,101) and (314.6,97.28) .. (314.6,92.68) -- cycle ;
        \draw  [fill={rgb, 255:red, 80; green, 227; blue, 194 }  ,fill opacity=1 ] (320.8,66.8) -- (329,66.8) -- (329,75) -- (320.8,75) -- cycle ;
        \draw  [fill={rgb, 255:red, 80; green, 227; blue, 194 }  ,fill opacity=1 ] (335.2,66.2) -- (343.4,66.2) -- (343.4,74.4) -- (335.2,74.4) -- cycle ;
        \draw  [fill={rgb, 255:red, 80; green, 227; blue, 194 }  ,fill opacity=1 ] (355.6,66.2) -- (363.8,66.2) -- (363.8,74.4) -- (355.6,74.4) -- cycle ;
        \draw  [fill={rgb, 255:red, 80; green, 227; blue, 194 }  ,fill opacity=1 ] (321.2,85.8) -- (329.4,85.8) -- (329.4,94) -- (321.2,94) -- cycle ;
        \draw  [fill={rgb, 255:red, 80; green, 227; blue, 194 }  ,fill opacity=1 ] (335.2,85.8) -- (343.4,85.8) -- (343.4,94) -- (335.2,94) -- cycle ;
        \draw  [fill={rgb, 255:red, 80; green, 227; blue, 194 }  ,fill opacity=1 ] (355.2,85.4) -- (363.4,85.4) -- (363.4,93.6) -- (355.2,93.6) -- cycle ;
        \draw    (350.2,71.4) -- (347.6,71.4) ;
        \draw    (350.2,90.4) -- (347.6,90.4) ;
        \draw    (326.2,80.4) -- (323.6,80.4) ;
        \draw    (340.2,80.4) -- (337.6,80.4) ;
        \draw    (360.2,80.4) -- (357.6,80.4) ;
        \draw    (350.2,80.4) -- (347.6,80.4) ;
        
        \draw [color={rgb, 255:red, 139; green, 87; blue, 42 }  ,draw opacity=1 ]   (158,149.83) -- (195.4,104.38) ;
        \draw [shift={(196.67,102.83)}, rotate = 129.44] [color={rgb, 255:red, 139; green, 87; blue, 42 }  ,draw opacity=1 ][line width=0.75]    (10.93,-3.29) .. controls (6.95,-1.4) and (3.31,-0.3) .. (0,0) .. controls (3.31,0.3) and (6.95,1.4) .. (10.93,3.29)   ;
        \draw [color={rgb, 255:red, 189; green, 16; blue, 224 }  ,draw opacity=1 ]   (158,149.83) -- (336.06,103.5) ;
        \draw [shift={(338,103)}, rotate = 165.42] [color={rgb, 255:red, 189; green, 16; blue, 224 }  ,draw opacity=1 ][line width=0.75]    (10.93,-3.29) .. controls (6.95,-1.4) and (3.31,-0.3) .. (0,0) .. controls (3.31,0.3) and (6.95,1.4) .. (10.93,3.29)   ;
        \draw [color={rgb, 255:red, 139; green, 87; blue, 42 }  ,draw opacity=1 ]   (198.33,102.67) -- (378.07,150.32) ;
        \draw [shift={(380,150.83)}, rotate = 194.85] [color={rgb, 255:red, 139; green, 87; blue, 42 }  ,draw opacity=1 ][line width=0.75]    (10.93,-3.29) .. controls (6.95,-1.4) and (3.31,-0.3) .. (0,0) .. controls (3.31,0.3) and (6.95,1.4) .. (10.93,3.29)   ;
        \draw [color={rgb, 255:red, 189; green, 16; blue, 224 }  ,draw opacity=1 ]   (340.33,102) -- (378.74,149.28) ;
        \draw [shift={(380,150.83)}, rotate = 230.91] [color={rgb, 255:red, 189; green, 16; blue, 224 }  ,draw opacity=1 ][line width=0.75]    (10.93,-3.29) .. controls (6.95,-1.4) and (3.31,-0.3) .. (0,0) .. controls (3.31,0.3) and (6.95,1.4) .. (10.93,3.29)   ;
        \draw [line width=2.25]  [dash pattern={on 2.53pt off 3.02pt}]  (233,80) -- (306,80) ;
        
        \draw (127,190.5) node [anchor=north west][inner sep=0.75pt]  [font=\small] [align=left] {Source};
        \draw (351.67,189.67) node [anchor=north west][inner sep=0.75pt]  [font=\small] [align=left] {Destination};
        \draw (155.5,118) node [anchor=north west][inner sep=0.75pt]  [font=\scriptsize]  {$h_{i}^{1}$};
        \draw (205.33,116) node [anchor=north west][inner sep=0.75pt]  [font=\scriptsize]  {$h_{i}^{R}$};
        \draw (241.8,96.8) node [anchor=north west][inner sep=0.75pt]  [font=\scriptsize]  {$g_{i}^{1}$};
        \draw (362.2,110.6) node [anchor=north west][inner sep=0.75pt]  [font=\scriptsize]  {$g_{i}^{R}$};
        \draw (181,39) node [anchor=north west][inner sep=0.75pt]  [font=\scriptsize]  {$RIS_{1}$};
        \draw (324,38.5) node [anchor=north west][inner sep=0.75pt]  [font=\scriptsize]  {$RIS_{R}$};

        \end{tikzpicture}
    \caption{System Model}
    \label{fig:sys}
    \end{figure}
     Let $h_{i}^r$  and $g_{i}^r$ represent the channel fading coefficients between the source to $i^{th}$ reflecting element of $r^{th}$ RIS, and $i^{th}$ reflecting element of $r^{th}$ RIS to destination respectively. Also, all the channels are assumed to undergo independent Rayleigh fading. Let $x$ be the transmitted signal, then the received signal at the destination reflected from $r^{th}$ RIS is given by
    \begin{equation}\label{one}
        y^{r}=\sqrt{P_{s}}\left [ \sum_{i=1}^{N_{r}}h_{i}^{r}\exp \left ( j\phi _{i}^r \right )g_{i}^{r} \right ]x+n^r,
    \end{equation}
    where $P_s$ is the source transmit power and $n^r$ is the additive white Gaussian noise (AWGN) with mean zero and variance $V_0$. Assume $d_{sr}$ and $d_{rd}$ are the distances between source to $r^{th}$ RIS and  $r^{th}$ RIS to destination, respectively. The small-scale fading channel gains are given by $h_{i}^{r}=\eta_{i}^{r}e^{-j\theta _{i}^{r}}$ and $g_{i}^{r}=\beta_{i}^{r}e^{-j\psi_{i}^{r}}$. Here  $\eta_{i}^{r}$, $\theta _{i}^{r}$ represent the channel amplitude and phase, respectively, for the link between source and $r^{th}$ RIS. Similarly, $\beta_{i}^{r}$, $\psi_{i}^{r}$  represent the channel amplitude and phase, respectively, for the link between $r^{th}$ RIS and destination.
    Similar to \cite{basar2019wireless} and using (\ref{one}), instantaneous SNR at destination from $r^{th}$ RIS is given by
    \begin{equation}
        \gamma ^{r}=\frac{P_{s}\left | \sum_{i=1}^{N_r}\eta_{i}^{r}\beta_{i}^{r}e^{j\left ( \phi_{i}^{r}-\theta _{i}^{r} -\psi _{i}^{r}\right )} \right |^{2}}{V_{0}}.
    \end{equation}
    Similar to \cite{basar2019wireless}, full channel state information is assumed to be available. So $\gamma^r$ can be maximized by setting the $\phi_{i}^{r}=\theta _{i}^{r} +\psi _{i}^{r}$. Therefore $\gamma^r$ can be written as 
    \begin{equation}
        \gamma ^{r}=\frac{P_{s}\left ( \sum_{i=1}^{N_r}\eta_{i}^{r}\beta_{i}^{r} \right )^{2}}{V_{0}}=\overline{\gamma }A_{r}^{2},
    \end{equation}
   where $A^{r}=\sum_{i=1}^{N_r}\eta_{i}^{r}\beta_{i}^{r}$ and  average SNR $\overline{\gamma }=\frac{P_{s}}{V_{0}}$.\\
   At the destination, the RIS with the highest SNR is selected for communication. Assuming $\overline{\gamma}$ is the average channel SNR in the past window \cite{yang2020outage}, the selection principle at the destination similar to \cite{yang2020outage} is given by
    \begin{equation}
       \hat{r}=\arg \max_{r=1,\cdots, R} \gamma ^{r},
   \end{equation}
   where $\gamma ^{r}=A_{r}^{2}$. As the number of reflecting surfaces in RIS becomes large $N_{r}\gg 1$, using central limit theorem (CLT) it is shown \cite{basar2019wireless},\cite{yang2020coverage} that  $A_r$ follows Gaussian distribution with mean $\frac{N_r \pi}{4}$ and variance of $N_r\left ( 1-\frac{\pi ^{2}}{16} \right )$ . Hence, we can see that $A_r^{2}$ will be an NCCS RV with one d.o.f. \\
  Considering a source, destination, and $R$ RISs in between, $R$ links with SNRs $\{\gamma^r\}_{r=1}^R $ are available at the destination. Each $\gamma^r$ follows a NCCS distribution with one d.o.f with the parameters $\lambda_r =\left ( \frac{N_r\pi }{4} \right )^{2}$ and $\sigma_r ^{2}=N_r\left ( 1-\frac{\pi ^{2}}{16} \right )$. Here, $\lambda_r$ represents the non-centrality parameter, and $\sigma_{r}^2$ is the variance of NCCS distribution. Typically, the link with the highest SNR is selected for communication.
  \begin{equation}
      \gamma_{max}^R=\gamma_{\hat{r}}=\max_{r=1,\cdots, R} \gamma ^{r}.
  \end{equation}
    Now let us see how we find the distribution of maximum SNR ($\gamma_{max}^R$). Observe that parameters ($\lambda_r$, $\sigma_r $) of NCCS RVs depend on the number of reflecting elements $(N_r)$ of $r^{th}$ RIS. So, if we consider an equal number of reflecting elements on each RIS, then $\gamma_{max}^R$ would be the maximum of $R$ i.i.d NCCS RVs with one d.o.f. This was the case studied in \cite{yang2020outage}.  Different numbers of reflecting elements will result in $\gamma_{max}^R$ being the maximum of i.n.i.d. NCCS RVs. The exact distribution of $\gamma_{max}^R$ can be written as
    $F_{\gamma _{max}^{R}}\left ( \gamma  \right )=\prod_{r=1}^{R}F_{\gamma^{r}}\left ( \gamma  \right )$ and the exact distribution of $\gamma_{max}^R$ will involve fairly complicated expressions whose complexity increasing with increasing $R$. Instead, we will utilize the EVT to characterize the asymptotic distribution of maximum order statistics of i.n.i.d. NCCS RVs with one d.o.f in the next section.

    \section{ maximum order statistics of i.n.i.d. NCCS RVs } \label{sec_kth_max}
     The general procedure in finding the asymptotic distribution of maximum order statistics for i.i.d. RVs involves finding the maximum domain of attraction of the common distribution function. However, in the case of i.n.i.d. RVs, additional requirements have to be met in order for the maxima to be a non-degenerate distribution. Finding appropriate normalizing constants which satisfies the additional requirements is fairly challenging.
     In this section, we derive the maximum order statistics of a sequence of i.n.i.d. NCCS RVs with one d.o.f using EVT. Considering the normalizing constants $a_R$ and $b_R$, first, we will derive the asymptotic distribution of the normalized maximum SNR $(\widetilde{\gamma}_{max})$ where $\widetilde{\gamma}_{max}=\lim_{R\to \infty }\frac{\gamma_{max}^{R}-b_R}{a_R}$. Once we have $\widetilde{\gamma}_{max}$, characterization of $\gamma_{max}^{R}$ is simple. We will introduce some of EVT's key results from \cite{barakat2013limit} for the general i.n.i.d. case in order to facilitate the understanding of our proofs.\\
     Let $\left\{ \gamma_1,\gamma_2,..,\gamma_R \right\}$ be a sequence of independent random variables with  $\gamma_r\sim F_{r}\left( \gamma \right)$ for $r=1,2,..,R$. If $\gamma_{max}^{R}=\max\left\{ \gamma_{r} \right\}_{r=1}^{R}$, then CDF of $\gamma_{max}^{R}$ can be written as
     \begin{equation}
         G_{max}^{R}\left( \gamma \right)=P\left( \gamma_{max}^{R}\le \gamma \right)=\prod_{r=1}^{R}F_{r}\left( \gamma \right).
     \end{equation}
    The following uniformity assumptions (UAs) are required for the analysis of asymptotic order statistics. The sequences of CDFs ${F_r(\gamma)}$ and normalizing constants $a_R$ and $b_R$  are said to satisfy the UAs for maximum vector $\gamma_{max}^R$ if 
     \begin{equation}\label{uua1}
         \max_{1\le r\le R}\left\{ 1-F_r\left( a_R\gamma+b_R \right) \right\}\to 0 \quad as  \quad R\to \infty, 
     \end{equation}
     for all $ a_R\gamma+b_R >\alpha\left( F_r \right)$ and  $\alpha\left( F_r \right):=\inf \{ \gamma:F_r(\gamma)> 0 \} > -\infty $. Also, for a fixed number $0< t\leq 1$ and each sequence of integers $\left\{m_R \right\}_{R}$ such that $m_R< R$, when $R\to \infty$, $m_R\to \infty $ and $\frac{m_R}{R} \to t$, we should have that
     \begin{equation}\label{uua2}
         \tilde{u}\left ( t,\gamma \right )=\lim _{R\to \infty } \sum_{r=1}^{m_R}\left ( 1-F_r\left ( a_R\gamma+b_R \right ) \right ),
     \end{equation}
     exists and is finite for all $0< t\leq 1$, whenever it is finite for $t=1$. With the UAs in (\ref{uua1}) and (\ref{uua2}), the authors of \cite{barakat2013limit} presented the following lemma for characterizing the asymptotic distribution of the maximum random variable for the general i.n.i.d. case.
     \begin{lemma}\label{thm_max}
         Under the UA (\ref{uua1}) and (\ref{uua2}), a non-degenerate CDF $\tilde{G}_{max}\left ( \gamma \right )$ is the asymptotic distribution of $\frac{\gamma_{max}^R-b_R}{a_R}$ i.e., $G_{max}^R\left ( a_R\gamma+b_R \right ) =\prod _{r=1}^R F_r\left ( a_R\gamma+b_R \right )\overset{D}{\rightarrow}\tilde{G}_{max}\left ( \gamma \right )$ as $R \to \infty$ where $\overset{D}{\rightarrow}$ stands for convergence in distribution, if and only if
         \begin{equation}\label{the1}
             \tilde{u}\left ( \gamma \right )=\tilde{u}\left (1, \gamma \right )=\lim _{R\to \infty }\sum_{r=1}^{R}\left ( 1-F_r\left ( a_R\gamma+b_R \right ) \right )< \infty .
         \end{equation}
         Moreover $\tilde{G}_{max}\left ( \gamma \right )$ should have the form $\tilde{G}_{max}\left ( \gamma \right )=e^{-\tilde{u}\left ( \gamma \right )}$ and either (i) $\log \tilde{G}_{max}\left ( \gamma \right )$ is concave or (ii) $\omega _{max}=\omega \left ( \tilde{G}_{max}\left ( \gamma \right ) \right )$ is finite and $\log  \tilde{G}_{max}\left ( \omega _{max}-e^{-\gamma} \right )$ is concave or (iii) $\alpha _{max} =\alpha \left ( \tilde{G}_{max}\left ( \gamma\right ) \right )$ is finite and $\log  \tilde{G}_{max}\left ( \alpha _{max}-e^{\gamma} \right )$ is concave where $\gamma>0$ in (ii) and (iii).
     \end{lemma}
     \begin{proof}
         Please refer to \cite{barakat2013limit} for the proof.
     \end{proof}
     We will make use of Lemma \ref{thm_max} in deriving the asymptotic distribution of the random variable $\widetilde{\gamma}_{max}$  by arriving at normalizing constants $a_R$ and $b_R$, satisfying the UAs and (\ref{the1}) for i.n.i.d. NCCS RVs. Once we arrive at $\widetilde{\gamma}_{max}$, using this, we can obtain the distribution of $\gamma_{max}^{R}$.

     Now, we will consider the system model presented in Section II to derive the asymptotic distribution of $\gamma_{max}^{R}$. Note that considering different numbers of reflecting elements in each RIS, we must deal with i.n.i.d. RVs $\{\gamma^r\}_{r=1}^R $. And, if we assume the same number of reflecting elements for all the RISs, we will get i.i.d. RVs $\{\gamma^r=\gamma\}_{r=1}^R $. 
     \subsection{i.n.i.d. case}
     Now, let $\{\gamma^r\}_{r=1}^R $ be a sequence of NCCS random variables with one d.o.f, then its CDF is
     \begin{equation}\label{nccs_cdf}
         F_{\gamma^r}(\gamma)=1-Q_{\frac{1}{2}}\left(\frac{\sqrt{\lambda_r}}{\sigma_r}, \frac{\sqrt{\gamma}}{\sigma_r}\right).
     \end{equation}
     In (\ref{nccs_cdf}), $Q_{\frac{1}{2}}(,.,)$ is the Marcum-Q function \cite{molisch2012wireless}  and $\lambda_r$ is the non-centrality parameter of the non-central chi-square RV.
     As mentioned in section II, $\lambda_r =\left ( \frac{N_r\pi }{4} \right )^{2}$ and $\sigma_r ^{2}=N_r\left ( 1-\frac{\pi ^{2}}{16} \right )$ depends only on the number of reflecting elements and both $\lambda_r$ and $\sigma_r$ will take maximum value at same index corresponding to RIS with maximum number of reflecting elements. Let $R$ be the total number of RVs. We will assume $\left(\lambda_r, \sigma_r\right)$ takes a finite set of values i.e. $\lambda_r \in\left\{\lambda_1, \lambda_2 \ldots \lambda_P\right\}$ for all $r \in\{1, \ldots R\}$ and
    $\sigma_r \in\left\{\sigma_1, \sigma_2 \ldots \sigma_P\right\}$  for all  $r \in\{1, \cdots R\}$. Define
    \begin{equation*}
        R_{i}=\sum_{r=1}^R \mathbb{I}_{\lambda_r=\lambda_i, \sigma_r=\sigma_i} \quad  1 \leq i \leq P , 
    \end{equation*}
    
    Where $\mathbb{I}_{\lambda _{r},\sigma _{r}}:=\left\{\begin{matrix}
    1 & if \lambda _{r}=\lambda _{i} , \sigma _{r}=\sigma _{i} \\ 
    0 & if \lambda _{r}\neq \lambda _{i} , \sigma _{r}\neq \sigma _{i}
    \end{matrix}\right.$. Here $R_{i}$ represents the number of times  pair $\left(\lambda_i, \sigma_i\right)$ occurs among  $R$ values. In the case of the multi-RIS system model presented in section II, the SNR follows the CDF in (\ref{nccs_cdf}). 

     \begin{theorem}\label{main_thm}
     The asymptotic CDF of normalized maximum $(\widetilde{\gamma}_{max})$ of a sequence of i.n.i.d. non-central chi-square random variables with one d.o.f as $R\to \infty$ is given by 
     \begin{equation}
         F_{\tilde{\gamma}_{max}}(\gamma)=\exp \left ( - \exp \left ( - \gamma  \right )  \right ),         
     \end{equation}
     for normalizing constants $a_R=\frac{ \tilde{\sigma}^2}{\epsilon}$ and $b_R=\frac{\tilde{\sigma}^2}{\epsilon}  \left[\log(\tilde{R})-c_1 \right]$.
     Here, $\left(\lambda_r, \sigma_r\right)$ takes a finite set of values i.e. $\lambda_r \in\left\{\lambda_1, \lambda_2 \ldots \lambda_P\right\}$ and  $\sigma_r \in\left\{\sigma_1, \sigma_2 \ldots \sigma_P\right\}$  for all  $r \in\{1, \cdots R\}$.
    Further, $R_{i}$ represents the number of times  pair $\left(\lambda_i, \sigma_i\right)$ occurs among  $R$ values. Let $\tilde{\sigma}$ be the largest among $\left\{\sigma_{1},\cdots,\sigma_{P}\right\}$, $\tilde{\lambda}$ be the largest among $\left\{\lambda_{1},\cdots,\lambda_{P}\right\}$ and $\tilde{R}$ to be the largest among $\left\{R_{1},\cdots,R_{P}\right\}$. Also, $\epsilon$ is the Chernoff parameter $(0< \epsilon < \frac{1}{2})$ and 
    $$c_1=\frac{-1}{\tilde{\theta}}\left[ \log\left( 1-2\epsilon  \right)^{-\frac{1}{2}}+\frac{\epsilon}{2\left( 1-2\epsilon \right)}\frac{\tilde{\lambda}}{\tilde{\sigma}^{2}}  \right].$$        
    \end{theorem}
    \begin{proof}
         Lemma \ref{thm_max} states that if 
         \begin{equation}\label{s1}
           \tilde{u}\left ( \gamma \right )=\lim _{R\to \infty }\sum_{r=1}^{R}\left ( 1-F_r\left ( a_R\gamma+b_R \right ) \right )< \infty,           
         \end{equation}
         for some normalizing constants $a_R$ and $b_R$ satisfying UAs, we can derive the distribution of $\widetilde{\gamma}_{max}$. Further, Mezlers \cite[Chapter 5]{haan2006extreme} give the following conditions on $a_R$ and $b_R$ such that  UA (\ref{uua1}) and (\ref{uua2}) are satisfied:
         \begin{align}\label{con1}
    		\mid \log a_R \mid + \mid b_R \mid \ \to \infty \ \text{as} \ R \to \infty,
    	\end{align}
		and
		\begin{equation}\label{con2}
    		\begin{array}{l}
    		\frac{a_{R+1}}{a_{R}} \rightarrow 1, \\
    		\frac{\left(b_{R+1}-b_{R}\right)}{a_{R}} \rightarrow 0.
    		\end{array}
    	\end{equation}
    We will derive the distribution of $\tilde{\gamma}_{max}$ by finding a $a_R$ and a $b_R$ such that (\ref{s1}),  (\ref{con1}) and (\ref{con2}) are satisfied in order for the UAs to hold. From \cite[(18)]{yang2020outage}, the asymptotic form of the generalized Marcum Q-function can be expressed as 
   \begin{equation}\label{asy}
      Q_{n}\left ( x ,y \right )\simeq  \left ( 1-2\epsilon  \right )^{-n} \exp \left ( -\epsilon y ^{2} \right )\exp \left ( \frac{n\epsilon x ^{2}}{1-2\epsilon} \right ).
   \end{equation}
   Here $y^{2}> n\left ( x^{2}+2 \right )$  and $\epsilon$ is the Chernoff parameter $(0< \epsilon < \frac{1}{2})$ with optimum value $\epsilon_{0}=\frac{1}{2}\left ( 1-\frac{n}{y ^{2}}-\frac{n}{y ^{2}}\sqrt{1+\frac{x ^{2}y ^{2}}{n}} \right )$ 
    Using (\ref{asy}) and (\ref{nccs_cdf}), we can rewrite $ \tilde{u}\left ( \gamma \right )$ in (\ref{s1}) as
    \begin{align}
      \tilde{u}\left( \gamma \right)= & \lim_{R\to \infty }\sum_{i=1}^{R} R_{i} \left( 1-2\epsilon  \right)^{-\frac{1}{2}}\exp\left( -\frac{\epsilon }{\sigma_{i}^{2}} \left( a_{R}\gamma+b_{R} \right)\right)  \exp\left( \frac{\epsilon}{2\left( 1-2\epsilon \right)}\frac{\lambda_{i}}{\sigma_{i}^{2}} \right)
    \end{align}
    \begin{align}\label{eq8}
        \tilde{u}\left( \gamma \right)= &\lim_{R\to \infty }\sum_{i=1}^{R} R_{i} \left( 1-2\epsilon  \right)^{-\frac{1}{2}}\exp\left( \frac{\epsilon}{2\left( 1-2\epsilon \right)}\frac{\lambda_{i}}{\sigma_{i}^{2}} \right) \exp\left( -\frac{\epsilon }{\sigma_{i}^{2}} \left( a_{R}\gamma \right)\right)\exp\left( -\frac{\epsilon }{\sigma_{i}^{2}} \left(b_{R} \right)\right)
    \end{align}
     Choose $\tilde{\sigma}$ to be the largest among $\left\{\sigma_{1},\cdots,\sigma_{P}\right\}$, $\tilde{\lambda}$ to be the largest among $\left\{\lambda_{1},\cdots,\lambda_{P}\right\}$ and $\tilde{R}$ to be the largest among $\left\{R_{1},\cdots,R_{P}\right\}$. Let us assume the following $a_R$ and $b_R$ values satisfying the conditions in (\ref{con1}) and (\ref{con2}) 
     \begin{align}\label{ar}
    		a_R=\frac{ \tilde{\sigma}^2}{\epsilon},
    \end{align}
    and
    \begin{equation}\label{br}
    	b_R=\frac{ \tilde{\sigma}^2}{\epsilon} \left[\log(\tilde{R})-c_1 \right],
    \end{equation}
    where $c_1$ is a constant and we assume $\tilde{R} \to \infty$ as $R \to \infty$. Note $a_R$ and $b_R$ satisfy (\ref{con1}) and (\ref{con2}). 
   For the above choice of normalizing constant, UA (\ref{uua1}) is satisfied as $b_R\to \infty$ as $R \to \infty$. For UA (\ref{uua2}) to be satisfied, we require $  \tilde{u}\left( \gamma \right)< \infty$ and $a_R$ and $b_R$ should satisfy (\ref{con1}) and (\ref{con2}). Substituting $a_R$ and $b_R$ in (\ref{eq8}),
    \begin{align}
        \tilde{u}\left( \gamma \right)= &\lim_{R\to \infty }\sum_{i=1}^{R} R_{i} \left( 1-2\epsilon  \right)^{-\frac{1}{2}}\exp\left( \frac{\epsilon}{2\left( 1-2\epsilon \right)}\frac{\lambda_{i}}{\sigma_{i}^{2}} \right) \nonumber \\ &
        \exp\left( -\frac{\epsilon }{\sigma_{i}^{2}} \left( \frac{\tilde{\sigma}^2}{\epsilon}\gamma \right)\right)
        \exp\left( -\frac{\epsilon }{\sigma_{i}^{2}} \left(\frac{\tilde{\sigma}^2}{\epsilon} \left[\log(\tilde{R})-c_1) \right] \right)\right).
    \end{align}
     Let $\left(\frac{\tilde{\sigma}}{\sigma_i}\right)^2=\theta_{i}$ , so  $\theta_i \in\left\{\theta_1, \theta_2 \ldots \theta_P\right\}$.      
     \begin{align}\label{comp1}
         \tilde{u}\left( \gamma \right)= &\lim_{R\to \infty }\sum_{i=1}^{R} R_{i} \left( 1-2\epsilon  \right)^{-\frac{1}{2}}\exp\left( \frac{\epsilon}{2\left( 1-2\epsilon \right)}\frac{\lambda_{i}}{\sigma_{i}^{2}} \right) \exp\left( - \theta_{i}  \gamma \right)\exp\left( - \theta_{i} \log(\tilde{R})+ \theta_{i}c_1 \right).
     \end{align}
     After rearranging (\ref{comp1}), we have,      
    \begin{align}\label{eq11}
        \tilde{u}\left( \gamma \right)=& \lim_{R\to \infty }\sum_{i=1}^{R} \exp\left( - \theta_{i}  \gamma \right)\frac{R_{i}}{\tilde{R}^{\theta_{i}}}  \underbrace{ \left( 1-2\epsilon  \right)^{-\frac{1}{2}}\exp\left( \frac{\epsilon}{2\left( 1-2\epsilon \right)}\frac{\lambda_{i}}{\sigma_{i}^{2}} \right)\exp\left(  \theta_{i}c_1 \right)}_{\text{Term-1}}
    \end{align}
    \begin{align}\label{eq21}
        \tilde{u}\left( \gamma \right)= &\lim_{R\to \infty }\sum_{i=1}^{R} \exp\left( - \theta_{i}  \gamma \right)\frac{R_{i}}{\tilde{R}^{\theta_{i}}} \underbrace{ \exp\left ( \log  (\left( 1-2\epsilon  \right)^{-\frac{1}{2}} + \frac{\epsilon}{2\left( 1-2\epsilon \right)}\frac{\lambda_{i}}{\sigma_{i}^{2}}+\theta_{i}c_1 \right )}_{\text{Term-1}}
    \end{align}
    Let us find the constant $c_1$ from Term-1, i.e. obtain $c_1$ such that
    \begin{equation*}
         \theta_{i}c_1=-\log\left( 1-2\epsilon  \right)^{-\frac{1}{2}}-\frac{\epsilon}{2\left( 1-2\epsilon \right)}\frac{\lambda_{i}}{\sigma_{i}^{2}}. 
    \end{equation*}
   Choose $c_1=\frac{-1}{\tilde{\theta}}\left[ \log\left( 1-2\epsilon  \right)^{-\frac{1}{2}}+\frac{\epsilon}{2\left( 1-2\epsilon \right)}\frac{\tilde{\lambda}}{\tilde{\sigma}^{2}}  \right]$, where $\tilde{\theta}=\min _{i=1,2,..,P} \theta _{i}$ so that Term-1 in (\ref{eq21}) will become one for $\tilde{\sigma}=\sigma_i$. Note that $\theta_i$ takes values greater than or equal to one, and when $\tilde{\sigma}=\sigma_i$ then only $\theta_i=1$. Substituting $c_1$ in (\ref{eq21})
    \begin{align}
        \tilde{u}\left( \gamma \right)=& \lim_{R\to \infty }\sum_{i=1}^{R} \exp\left( - \theta_{i}  \gamma \right)\frac{R_{i}}{\tilde{R}^{\theta_{i}}}  \exp\left( \log\left( 1-2\epsilon  \right)^{-\frac{1}{2}}+\frac{\epsilon}{2\left( 1-2\epsilon \right)}\frac{\lambda_{i}}{\sigma_{i}^{2}}\right)  \nonumber \\ &\exp \left(-\frac{\theta_{i}}{\tilde{\theta}}\left[ \log\left( 1-2\epsilon  \right)^{-\frac{1}{2}}+\frac{\epsilon}{2\left( 1-2\epsilon \right)}\frac{\tilde{\lambda}}{\tilde{\sigma}^{2}} \right] \right).
   \end{align}
   Therefore, 
   \begin{equation}\label{imp}
        \tilde{u}\left ( \gamma \right )=\sum_{i=1}^{P}\left ( \exp \left ( - \theta_{i}\gamma \right )p_i\right ),
   \end{equation}
   where $p_i=\frac{R_{i}}{\tilde{R}^{\theta_{i}}}\exp\left( \log\left( 1-2\epsilon  \right)^{-\frac{1}{2}}+\frac{\epsilon}{2\left( 1-2\epsilon \right)}\frac{\lambda_{i}}{\sigma_{i}^{2}}\right)\exp\left(-\frac{\theta_{i}}{\tilde{\theta}}\left[ \log\left( 1-2\epsilon  \right)^{-\frac{1}{2}}+\frac{\epsilon}{2\left( 1-2\epsilon \right)}\frac{\tilde{\lambda}}{\tilde{\sigma}^{2}} \right] \right)$.
   Note that $\tilde{R}\rightarrow \infty $ when $R\rightarrow \infty $ and there are $P$ values for $p_i$ i.e., $i=1,2,.,P$. Only for one $i$ $( \lambda_i=\tilde{\lambda}, \sigma_i=\tilde{\sigma})$, $\theta_i=1$ and in that scenario $p_i=1$ because $R_i=\tilde{R}$ and all the terms in the exponent go to zero.  Also, for this case $ \tilde{u}\left( \gamma \right)$ is finite as $\lim_{R\to \infty }\frac{R_{i}}{\tilde{R}^{\theta_{i}}}=1$. For all remaining i's $(\lambda_i \neq\tilde{\lambda},\sigma_i\neq\tilde{\sigma})$ as $R_i< \tilde{R}$ and $\theta_i>1$ corresponding $p_i=0$ as  $\lim_{R\to \infty }\frac{R_{i}}{\tilde{R}^{\theta_{i}}}=0$. Hence, the summation in (\ref{imp}) is finite making $\tilde{u}\left ( \gamma \right )=\exp(-\gamma) < \infty$ for the choice of normalizing constants $ a_R=\frac{ \tilde{\sigma}^2}{\epsilon}$ and $ b_R=\frac{ \tilde{\sigma}^2}{\epsilon} \left[\log(\tilde{R})-c_1 \right]$. Furthermore, from  Lemma \ref{thm_max} $\tilde{G}_{max}\left ( \gamma \right )=e^{-\tilde{u}\left ( \gamma \right )}=\exp(-\exp(-\gamma))$ and note that $\log \tilde{G}_{max}\left ( \gamma \right )=-\exp(-\gamma)$ is concave.
   Hence, the asymptotic CDF of the normalized maximum of a non-central chi-square RVs with one d.o.f is given by
   \begin{equation}\label{inid_cdf}
         F_{\tilde{\gamma}_{max}}(\gamma)=\exp \left ( - \exp \left ( - \gamma  \right )  \right ).        
     \end{equation}
   Note that for all practical purposes (finite values of $R$ and $\tilde{R}$), one can still use  $ \tilde{u}\left ( \gamma \right )=\sum_{i=1}^{P}\left ( \exp \left ( - \theta_{i}\gamma \right )p_i\right )$ and   
   \begin{equation}\label{cdf}
        F_{\tilde{\gamma}_{max}}(\gamma)=\exp \left ( -\sum_{i=1}^{P}\left ( \exp \left ( - \theta_{i}\gamma \right ) p_i\right ) \right ),
   \end{equation}
    where $p_i=\frac{R_{i}}{\tilde{R}^{\theta_{i}}}\exp\left( \log\left( 1-2\epsilon  \right)^{-\frac{1}{2}}+\frac{\epsilon}{2\left( 1-2\epsilon \right)}\frac{\lambda_{i}}{\sigma_{i}^{2}}\right)\exp\left(-\frac{\theta_{i}}{\tilde{\theta}}\left[ \log\left( 1-2\epsilon  \right)^{-\frac{1}{2}}+\frac{\epsilon}{2\left( 1-2\epsilon \right)}\frac{\tilde{\lambda}}{\tilde{\sigma}^{2}} \right] \right)$.
  \end{proof}  
  \subsection{i.i.d. case}
  If all the RISs have the same number of reflecting elements ($N_r=N$), then each $\gamma^r$ follows an NCCS distribution with one d.o.f with the parameters $\lambda_r =\left ( \frac{N\pi }{4} \right )^{2}$ and $\sigma_r ^{2}=N\left ( 1-\frac{\pi ^{2}}{16} \right )$. So $\gamma _{max}=\max _{r=1,2,..,P}\left \{ \gamma ^{r}=\gamma  \right \}$ and distribution of normalized maximum $F_{\tilde{\gamma}_{max}}(\gamma)$ can be derived as a special case of Theorem \ref{main_thm} and is presented in the following corollary.
  \begin{corollary}\label{cor}
       The asymptotic CDF of $\widetilde{\gamma}_{max}$ of a sequence of i.i.d. non-central chi-square random variables with one d.o.f as $R\to \infty$ is given by 
     \begin{equation}\label{iid_cdf}
        F_{\tilde{\gamma}_{max}}(\gamma)=\exp\left [ - \exp\left ( - \gamma \right ) \right ],        
     \end{equation}
     for normalizing constants $a_R=\frac{\sigma^2}{\epsilon}$ and $b_R=\frac{\sigma^2}{\epsilon}  \left[\log(R)-c_1 \right]$.    
     Also, $\epsilon$ is the Chernoff parameter $(0< \epsilon < \frac{1}{2})$ and $c_1=-\left[ \log\left( 1-2\epsilon  \right)^{-\frac{1}{2}}+\frac{\epsilon}{2\left( 1-2\epsilon \right)}\frac{\lambda}{\sigma^{2}}  \right]$ are constants.        
  \end{corollary}
  \begin{proof}
      This result can be derived by substituting $\lambda_r=\lambda$ and $\sigma_r=\sigma$ for all $r=\{1,2,..,P\}$ in Theorem \ref{main_thm}.
  \end{proof}
  For all practical cases one can obtain the unnormalized statistics by substituting $\gamma$ by $\frac{\gamma-b_R}{a_R}$ in (\ref{inid_cdf}) . Therefore $F_{\gamma_{max}}(\gamma)$ can be written as
  \begin{equation}
      F_{\gamma _{max}}\left ( \gamma  \right )=F_{\tilde{\gamma }_{max}}\left ( \gamma  \right )|_{\gamma =\frac{\gamma -b_R}{a_R}}.
  \end{equation}
\subsection{$k$-th order statistics}
So far, we have analyzed the multi-RIS system in a scenario where the link with the highest SNR is selected. If we are interested in choosing the link with $k$-th highest SNR instead of maximum SNR, we require the $k$-th order statistics, i.e., we want to find the normalized $k$-th maximum distribution of NCCS random variables with one d.o.f.  Let $\gamma_{ (1: R)} \leq \gamma_{(2: R)} \leq \cdots \leq\gamma_{(R:R)}$, be the order statistics where the $k$-th order statistic is given by $\gamma_{(R-k+1:R)}$. Finding the exact CDF of $k$-th order statistic $\gamma_{(R-k+1:R)}$  involves very complicated expression as given  \cite[(5.2.1)]{david2003order} 
\begin{align}
& F_{\gamma_{(R-k+1:R)}}(\gamma)=\sum_{m=k}^{R} \sum_{S_{m}} \prod_{r=1}^{m} F_{\gamma_{j_r}}(\gamma)  \times  \prod_{r=m+1}^{R}\left[1-F_{\gamma_{j_r}}(\gamma)\right], \ k=1,2,\cdots,R,
\label{cdf_kth_max_exact}
\end{align}
where the summation $S_{m}$ is over all the permutations $\left(j_{1}, \ldots, j_{R}\right)$ of $1, \ldots, R$ for which $j_{1}<\cdots<j_{m}$ and $j_{m+1}<\cdots<j_{R}$. A simpler asymptotic CDF of $\gamma_{(R-k+1:R)}$ can be computed with the help of EVT.
Asymptotic order statistics can be derived with the help of the results presented in \cite{barakat2002limit}. For ease of understanding, we will repeat an important Lemma here
   \begin{lemma} \label{thm_order_stat} \color{black}
        Assume that for suitable normalizing constants $a_{R}>0,$ $b_{R}$
        \begin{equation}
        \delta_{R}=\max _{1 \leq r \leq R} 1-F_{\gamma_r} \left(a_R\gamma+b_R\right) \rightarrow 0 \text { as } R \rightarrow \infty.
        \label{ua1}
        \end{equation}
        Then $\tilde{\phi}_{k:R}(\gamma) = \mathbb{P}\left(\frac{\gamma_{(R-k+1:R)}-b_R}{a_R} \leq \gamma \right)$ converges weakly to a non degenerate distribution function $\tilde{\phi}_{k}\left(\gamma\right)$ if and only if, for all $\gamma$ for which $\tilde{\phi}_{k}\left(\gamma\right)>0$, the limit
        \begin{equation}
        \tilde{u}(\gamma) = \lim _{R \rightarrow \infty} \sum_{r=1}^{R} 1-F_{\gamma_r}(a_R\gamma+b_R) \ \text{is finite,}
        \label{ua2}
        \end{equation}
        and the function
        \begin{equation}\label{kth}
        \tilde{\phi}_{k}\left(\gamma\right) = \sum\limits_{r=0}^{k-1} \frac{\tilde{u}^r(\gamma)}{r!} \exp (-\tilde{u}(\gamma)), \ \text{is a non degenerate distribution.}
        \end{equation}
        The actual limit of $\tilde{\gamma}_{(R-k+1:R)}=\frac{\gamma_{(R-k+1: R)}-b_{R}}{a_{R}} $ is the one given in (\ref{kth}). 
    \end{lemma}
    	
    \begin{proof}
    Please refer \cite{barakat2002limit} for the detailed proof.
    \end{proof} 

    Now, we use Lemma \ref{thm_order_stat} to derive the $k$-th order statistics of i.n.i.d. NCCS random variables with one d.o.f for the system model presented in section II. We can observe that in $(\ref{imp})$, we have already proved that $\tilde{u}(\gamma)$  is finite for the normalizing constants $a_R=\frac{ \tilde{\sigma}^2}{\epsilon}$ and $b_R=\frac{\tilde{\sigma}^2}{\epsilon}  \left[\log(\tilde{R})-c_1 \right]$ satisfying the UA (\ref{uua1}) and (\ref{uua2}). For $k$-th order statistics also, we need to satisfy the equations (\ref{ua1}) and (\ref{ua2}), which are same as UA (\ref{uua1}) and (\ref{uua2}). Hence, we can utilize the obtained $\tilde{u}(\gamma)$ from (\ref{inid_cdf}).  So, if we substitute the derived $\tilde{u}(\gamma)$ in $(\ref{kth})$ and show that $\tilde{\phi}_{k}$ is a non-degenerate distribution, then we obtain the normalized distribution of $k$-th order statistics. We present the results in the following corollary.
    \begin{corollary}\label{kthmax}
       The asymptotic CDF of normalized $k$-th maximum of a sequence of i.n.i.d. non-central chi-square random variables with one d.o.f as $R\to \infty$ is given by, 
       \begin{align}\label{kth_1}
        \tilde{\phi}_{k}\left(\gamma\right) = &\sum\limits_{r=0}^{k-1} \frac{\left ( \tilde{u}(\gamma)\right )^r}{r!} \exp (- \tilde{u}(\gamma))=F_{\tilde{\gamma} _{(R-k+1:R)}}\left ( \gamma \right ) =\frac{\Gamma \left ( k, \tilde{u}(\gamma)\right )}{\Gamma \left ( k \right )},
        \end{align}
        where
        \begin{equation}\label{imp_1}
        \tilde{u}\left ( \gamma \right )=\exp \left ( - \gamma \right ),
        \end{equation}
       for normalizing constants $a_R=\frac{ \tilde{\sigma}^2}{\epsilon}$ and $b_R=\frac{\tilde{\sigma}^2}{\epsilon}  \left[\log(\tilde{R})-c_1 \right]$.
        Assume $\left(\lambda_r, \sigma_r\right)$ takes a finite set of values i.e. $\lambda_r \in\left\{\lambda_1, \lambda_2 \ldots \lambda_P\right\}$ and  $\sigma_r \in\left\{\sigma_1, \sigma_2 \ldots \sigma_P\right\}$  for all  $r \in\{1, \cdots R\}$.
        Further, $R_{i}$ represents the number of times  pair $\left(\lambda_i, \sigma_i\right)$ occurs among  $R$ values. Let $\tilde{\sigma}$ be the largest among $\left\{\sigma_{1},\cdots,\sigma_{P}\right\}$ and $\tilde{\lambda}$  be the largest among $\left\{\lambda_{1},\cdots,\lambda_{P}\right\}$ and also $\tilde{R}$  be the largest among $\left\{R_{1},\cdots,R_{P}\right\}$. Also, $\epsilon$ is the Chernoff parameter $(0< \epsilon < \frac{1}{2})$ and $c_1=\frac{-1}{\tilde{\theta}}\left[ \log\left( 1-2\epsilon  \right)^{-\frac{1}{2}}+\frac{\epsilon}{2\left( 1-2\epsilon \right)}\frac{\tilde{\lambda}}{\tilde{\sigma}^{2}}  \right]$.    
     \end{corollary}
     \begin{proof}
    The results can be derived by substituting $\tilde{u}(\gamma)$ derived in (\ref{inid_cdf}) into (\ref{kth}) 
     \begin{equation}\label{kth_2}
        \tilde{\phi}_{k}\left(\gamma\right) = \sum\limits_{r=0}^{k-1} \frac{\left (\exp\left ( - \gamma \right )\right )^r}{r!} \exp (-\exp\left ( - \gamma \right )).
    \end{equation}
    The upper incomplete gamma function for integer $k$ \cite{upper_gamma} can be written as
    $$
   \Gamma\left ( k,x \right )=\left ( k-1 \right )!e^{-x}\sum_{r=0}^{k-1}\frac{x^{r}}{r!}.
    $$
    Using the $\Gamma \left ( k,x \right )$ and $\Gamma (x)=(x-1)! $ for integer $x$ \cite{gamma}, we can rewrite (\ref{kth_2}) as
     \begin{equation}\label{kth_3}
        \tilde{\phi}_{k}\left(\gamma\right) = F_{\tilde{\gamma} _{(R-k+1:R)}}\left ( \gamma \right )=\frac{\Gamma \left ( k, \tilde{u}(\gamma)\right )}{\Gamma \left ( k \right )}.
    \end{equation}
   Note that now, $\tilde{\phi}_{k}\left(\gamma\right) $ is a non degenerate function as $\frac{\Gamma \left ( k, \tilde{u}(\gamma)\right )}{\Gamma \left ( k \right )}$ is not a one point distribution. 
    
    \end{proof}
Here also, we can note that for all practical purposes, one can use $ \tilde{u}\left ( \gamma \right )=\sum_{i=1}^{P}\left ( \exp \left ( - \theta_{i}\gamma \right )p_i\right )$ with 
$p_i=\frac{R_{i}}{\tilde{R}^{\theta_{i}}}\exp\left( \log\left( 1-2\epsilon  \right)^{-\frac{1}{2}}+\frac{\epsilon}{2\left( 1-2\epsilon \right)}\frac{\lambda_{i}}{\sigma_{i}^{2}}\right)\exp\left(-\frac{\theta_{i}}{\tilde{\theta}}\left[ \log\left( 1-2\epsilon  \right)^{-\frac{1}{2}}+\frac{\epsilon}{2\left( 1-2\epsilon \right)}\frac{\tilde{\lambda}}{\tilde{\sigma}^{2}} \right] \right).$
We can observe that once we have the results for normalized $k$-th order statistics, we can derive the results for unnormalized $k$-th order statistics by substituting $\gamma$ by $\frac{\gamma-b_R}{a_R}$ in (\ref{kth_1}).
\subsection{Stochastic ordering of $k$-th maximum SNR}
The CDF of $k$-th maximum RV in terms of normalizing constants can be written as
$$
F_{\tilde{\gamma} _{(R-k+1:R)}}\left ( \gamma \right )=\frac{\Gamma \left ( k, \tilde{u}\left ( \frac{\gamma -b_R}{a_R} \right )\right )}{\Gamma \left ( k \right )}.
$$
Let $A$ and $B$ are two $k$-th maximum RVs with normalizing constants $(a_{R}^{A}, b_{R}^{A})$ and $(a_{R}^{B}, b_{R}^{B})$ respectively.
From stochastic ordering, an RV $A$ is stochastically smaller than RV $B$ if \cite{subhash2022asymptotic}
\begin{equation}
   \mathbb{P}\left ( A> z \right )\leq \mathbb{P}\left ( B> z \right ) \quad \forall   z\in\mathbb{R}.
\end{equation}
We can write the same in terms of CDF expressions in the following form
\begin{equation}\label{ineq}
    \Gamma \left ( k, \tilde{u}\left ( \frac{\gamma -b_{R}^B}{a_{R}^B} \right )\right ) \leq \Gamma \left ( k, \tilde{u}\left ( \frac{\gamma -b_{R}^A}{a_{R}^A} \right )\right ).
\end{equation}
As $\Gamma \left ( k,x \right )$ is a decreasing function with respect to its argument $x$, the inequality (\ref{ineq} ) is true if $\tilde{u}\left ( \frac{\gamma -b_{R}^B}{a_{R}^B} \right )\geq \tilde{u}\left ( \frac{\gamma -b_{R}^A}{a_{R}^A} \right )$, i.e.
$$
\exp\left( \frac{-\gamma+b_{R}^B}{a_{R}^B}\right) \geq \exp\left( \frac{-\gamma+b_{R}^A}{a_{R}^A}\right).
$$
We can observe that closed-form expressions of normalizing constants have a one-to-one mapping with the parameters, as shown below
\begin{align*}
&a_R^{I}=\frac{ \tilde{\sigma_{I}}^2}{\epsilon},  \nonumber \\ &
b_{R}^{I}=\frac{\tilde{\sigma}_{I}^2}{\epsilon}  \left[\log(\tilde{R_{I}})+\frac{1}{\tilde{\theta_{I}}}\left( \log\left( 1-2\epsilon  \right)^{-\frac{1}{2}}+\frac{\epsilon}{2\left( 1-2\epsilon \right)}\frac{\tilde{\lambda_{I}}}{\tilde{\sigma_{I}}^{2}}  \right) \right],
\end{align*}
where $I\in \left \{ A, B \right \}$. Here, any parameter changes are reflected in corresponding normalizing constants. Hence, stochastic ordering can be established as the expressions for normalizing constants have the corresponding mapping with the parameters.
\subsection{Average Throughput and Outage Capacity}
In this subsection, we will derive the expression for multiple RIS-aided communication systems' average throughput and outage capacity. Here, we consider the received SNR at the destination as the $k$-th maximum order statistics following the CDF expression presented in  (\ref{kth_1}). Also, here we will consider the SNR random variable to be of type $\gamma_a\gamma_{(R-k+1:R)}$ where $\gamma_a$ is a constant.
\subsubsection{Average Throughput }
Given the CDF of normalized $k$-th maximum $(\gamma_{(R-k+1:R)})$ of i.n.i.d. NCCS RVs with one d.o.f (\ref{kth_1}), we can derive the average throughput at the receiver. The expression for average throughput can be written as 
\begin{equation}\label{er_1}
    C_{R-k+1:R}=\mathbb{E}\left [\log_2\left ( 1+\gamma_a\gamma_{(R-k+1:R)} \right )  \right ].
\end{equation}
The expression in (\ref{er_1}) can be solved using the following numerical integration.
\begin{equation}\label{er_2}
    C_{R-k+1:R}=\int_{0}^{\infty }\log_2\left ( 1+\gamma_a\gamma_{(R-k+1:R)} \right )f_{\gamma_{(R-k+1:R)}}\left ( \gamma  \right )d\gamma.
\end{equation}
Here, $f_{\gamma_{(R-k+1:R)}}\left ( \gamma  \right )$ is the pdf of $k$-th maximum  i.e. $\gamma_{(R-k+1:R)}$. The pdf expression of $f_{\gamma_{(R-k+1:R)}}\left ( \gamma  \right )$ can be obtained by differentiating the CDF expression in (\ref{kth_1}) and can be written as
\begin{align}\label{appr}
    f_{\gamma_{(R-k+1:R)}}\left ( \gamma  \right )= &\frac{1 }{a_R\Gamma \left ( k \right )}\exp\left (- \exp \left ( -\frac{\gamma -b_{R}}{a_{R}} \right ) \right )\left ( \exp \left ( - \frac{\gamma -b_{R}}{a_{R}} \right ) \right )^{k}.
\end{align}
Hence average throughput can be calculated by substituting $f_{\gamma_{(R-k+1:R)}}\left ( \gamma  \right )$ in (\ref{er_2}) with the help of numerical integration routines. 
\subsubsection{Outage Capacity}
Given the CDF of $k$-th maximum (\ref{kth_1}), the outage probability for a threshold $\gamma_{th}$ can be calculated as 
\begin{align}
     P_{out}= & \mathbb{P}\left ( \gamma _{a} \gamma_{(R-k+1:R)}\leq \gamma _{th}\right )=F_{\gamma _{(R-k+1:R)}}\left ( \frac{\gamma _{th}}{\gamma _{a}} \right )=\frac{\Gamma \left ( k, \tilde{u}(\frac{\frac{\gamma_{th}}{\gamma_a} -b_{R}}{a_{R}})\right )}{\Gamma \left ( k \right )}.
\end{align}
Similarly, Outage capacity can be calculated as 
\begin{equation}\label{c_out}
    C_{out}=\log _{2}\left ( 1+\gamma _{th} \right )\left ( 1-F_{\gamma _{(R-k+1:R)}}\left (\gamma _{th}  \right ) \right ).
\end{equation}
We have derived the asymptotic distribution of the maximum of R NCCS RVs with one d.o.f. for the cases of a). i.n.i.d.  b). i.i.d. RVs. Further, we have derived the $k$-th order statistics of NCCS RVs and presented stochastic ordering results. The derived asymptotic distributions are used to find the multi-RIS system's outage capacity and average throughput. In the next section, we will see how these asymptotic results serve as approximations in simulations even when $R$ is not tending to infinity.
\section{Simulation Results} 
 In the simulations, we consider R RISs, and each RIS takes the number of reflecting elements from a finite set ${N_1, N_2, N_3}$. As per the system model, we can observe that received SNR $\gamma^r$ at the destination follows an NCCS distribution with one d.o.f with the parameters $\lambda_r =\left ( \frac{N_r\pi }{4} \right )^{2}$ and $\sigma_r ^{2}=N_r\left ( 1-\frac{\pi ^{2}}{16} \right )$.
 \begin{figure}[h]
    \centering
    \includegraphics[scale=0.5]{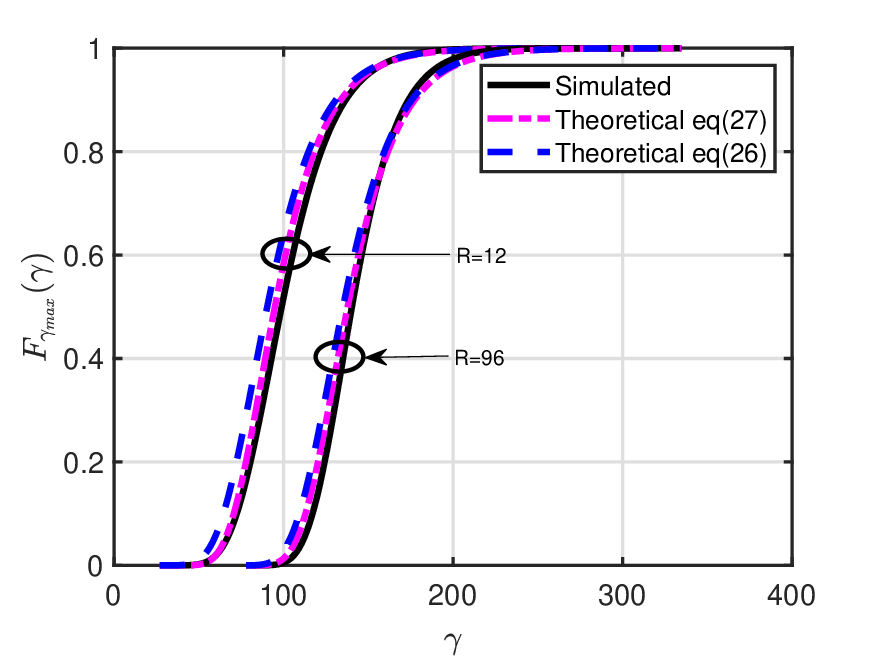}
    \caption{ CDF of $\gamma_{max}$ (equations  (\ref{inid_cdf}) ,(\ref{cdf}))  for i.n.i.d. RVs with $N_1=10$, $N_2=8$, and $N_3=6$}
    \label{fig:inid_comp}
\end{figure}
 In the simulations, we consider $R$ i.n.i.d NCCS RVs with one d.o.f with the CDF given in (\ref{nccs_cdf}). Here we consider $R_i$ RVs with parameters $(\lambda_i, \sigma_i)$ where $i \in \left \{ 1,\cdots, P \right \}$ such that $ \sum_{i=1}^{P} R_{i}=R$. We compare the theoretical and empirical CDFs of maximum order statistics. In Fig. \ref{fig:inid_comp}, we compare equation (\ref{inid_cdf}) with equation (\ref{cdf}) and simulated CDF. Note, for finite $R$ our derived equation (\ref{cdf}) will always be close to simulated CDF. Hence, for all other figures, we compare the simulated CDF with the derived equation (\ref{cdf}) when we discuss them.
\subsection{i.n.i.d. results}
\begin{figure}[h]
    \centering
    \includegraphics[scale=0.5]{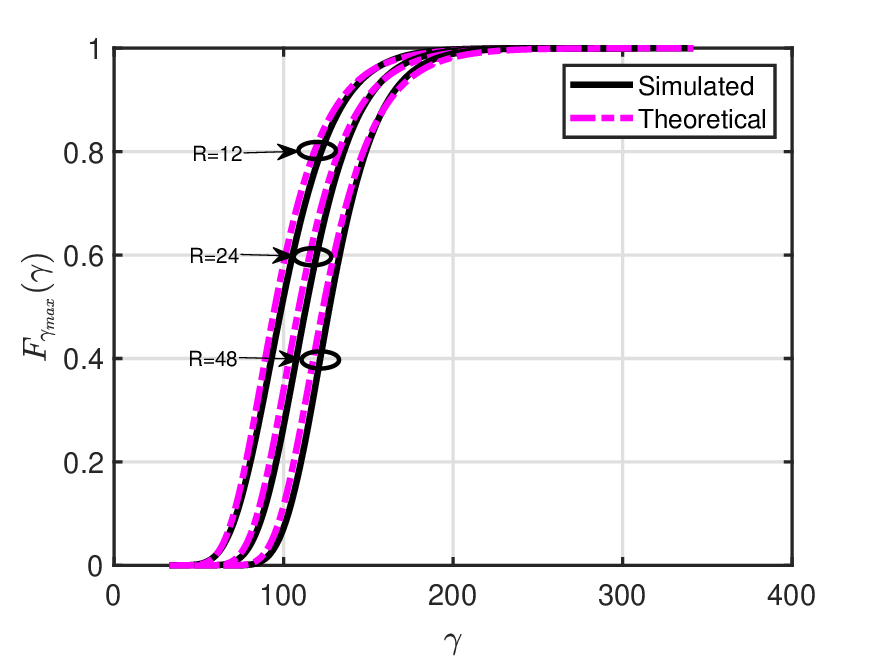}
    \caption{CDF of $\gamma_{max}$ for i.n.i.d. RVs with $N_1=10$, $N_2=8$, and $N_3=6$}
    \label{fig:inid_1}
\end{figure}
\begin{figure}[h]
    \centering
    \includegraphics[scale=0.5]{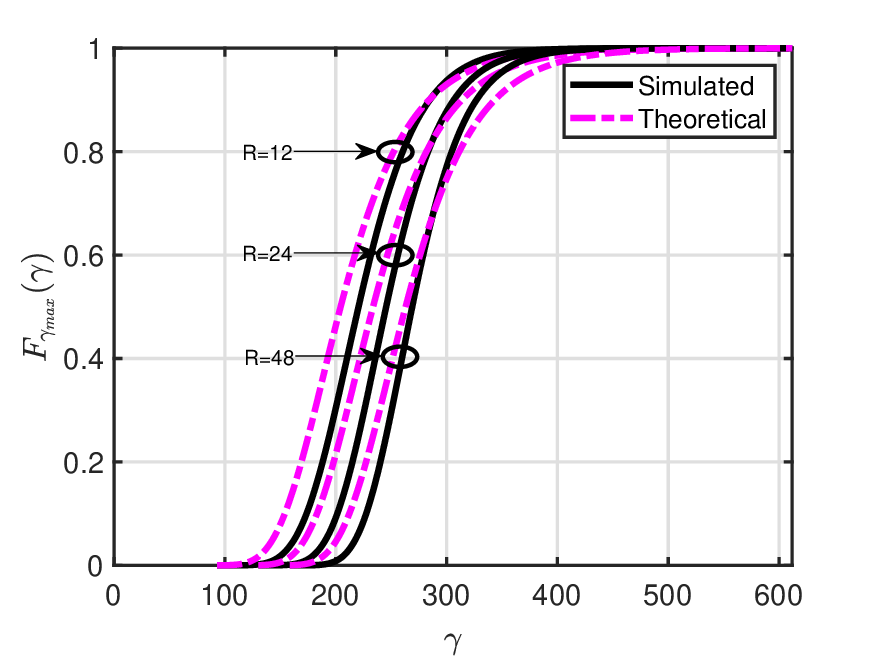}
    \caption{CDF of $\gamma_{max}$ for i.n.i.d. RVs with $N_1=15$, $N_2=13$, and $N_3=11$}
    \label{fig:inid_2}
\end{figure}
In Fig. \ref{fig:inid_1}, we present the CDFs of maximum order statistics for the values of $R=12$, $24$, and $48$ with the following number of reflecting elements $N_1=10$, $N_2=8$, and $N_3=6$ with $R_1=\frac{R}{3}$, $R_2=\frac{R}{3}$, and $R_3=\frac{R}{3}$.
We assume $\epsilon =\sqrt{\frac{\tilde{\sigma }}{\tilde{\lambda }}}$ as the constant throughout the simulations. Fig. \ref{fig:inid_2} presents the results assuming $N_1=15$, $N_2=13$, and $N_3=11$ as the number of reflecting elements. In Fig. \ref{fig:inid_2} we consider the values $R=12$, $24$, and $48$  with $R_1=\frac{R}{2}$, $R_2=\frac{R}{4}$, and $R_3=\frac{R}{4}$.
In Fig. \ref{fig:inid_1} and Fig. \ref{fig:inid_2}, the solid line presents the results of simulated CDF, and the dashed line presents the results of the theoretical CDF of maximum order statistics. Here, we can observe that, in both cases (Fig. \ref{fig:inid_1},\ref{fig:inid_2}), simulated and theoretical values are close to each other, and we can also observe that as the number of RVs ($R$) increases, we are getting better results. Also, from the simulations, we can observe that even for small values of $R$, the derived results are in good agreement with the simulated results. The results of Fig. \ref{fig:inid_1} are better than Fig. \ref{fig:inid_2} as the approximation of the Marcum-Q function in (\ref{asy}) is good for lower values of $N$.

\subsection{i.i.d.results}
\begin{figure}[h]
    \centering
    \includegraphics[scale=0.5]{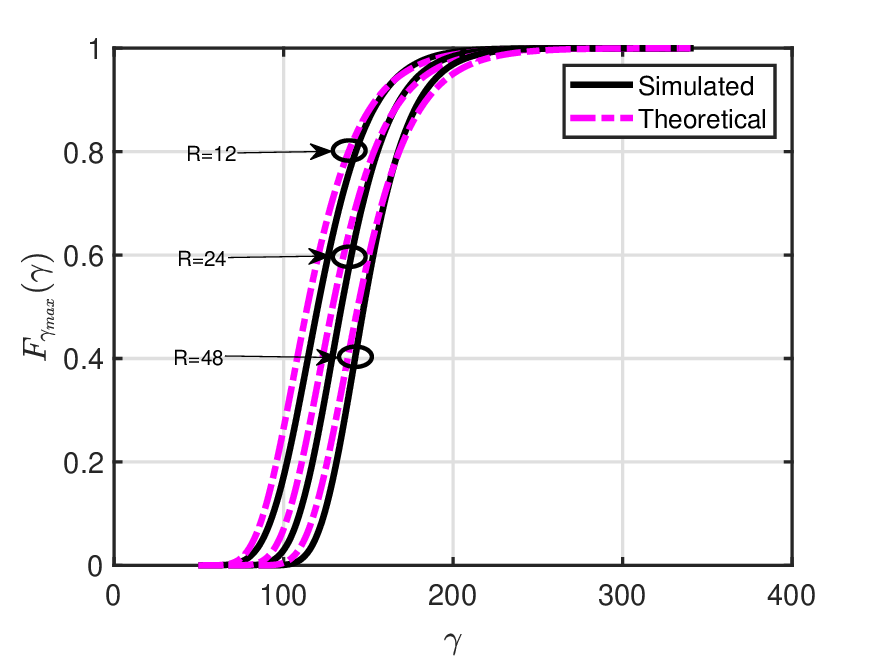}
    \caption{CDF of $\gamma_{max}$ for i.i.d. RVs with $N=10$}
    \label{fig:iid_1}
\end{figure}
 \begin{figure}[h]
    \centering
    \includegraphics[scale=0.5]{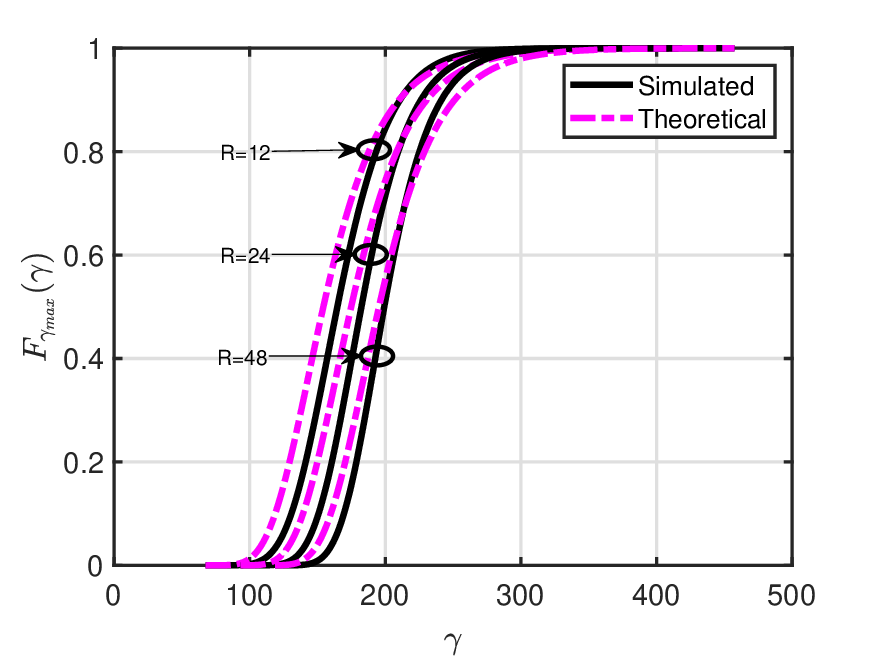}
    \caption{CDF of $\gamma_{max}$ for i.i.d. RVs with $N=12$}
    \label{fig:iid_2}
\end{figure}
 In this section, we validate the results of the section-III corollary \ref{cor}. Here we consider the equal number of reflecting elements in each RIS leading to i.i.d. NCCS RVs with one d.o.f. We assume that all RISs will have the same number of reflecting elements, i.e., $N$, and we consider $R$ i.i.d. NCCS random variables.
Fig. \ref{fig:iid_1} presents the results assuming $N=10$ as the number of reflecting elements for all RISs. In Fig. \ref{fig:iid_1} we consider $R=12$, $24$, and $48$ respectively with $N_1=10$; $1\leq r\leq R $.
Fig. \ref{fig:iid_2} assumes $N=12$ as the number of reflecting elements for all RISs. In Fig. \ref{fig:iid_2} we consider $R=12$, $24$, and $48$ respectively with $N_1=12$; $1\leq r\leq R $. Here also, we can observe that theoretical and empirical CDFs for different values of $R$ are in good agreement, and there is improvement in the results with an increase in $R$. Note all prior work assumes the same number of reflecting elements for all RISs while using EVT to compute the maximum. In \cite{yang2020outage}, the authors assumed i.i.d. NCCS RVs and expressed the limiting distribution as Gumbel CDF. We derived the identical results as a special case and presented them in corollary \ref{cor}.
\begin{figure}[h]
    \centering
    \includegraphics[scale=0.5]{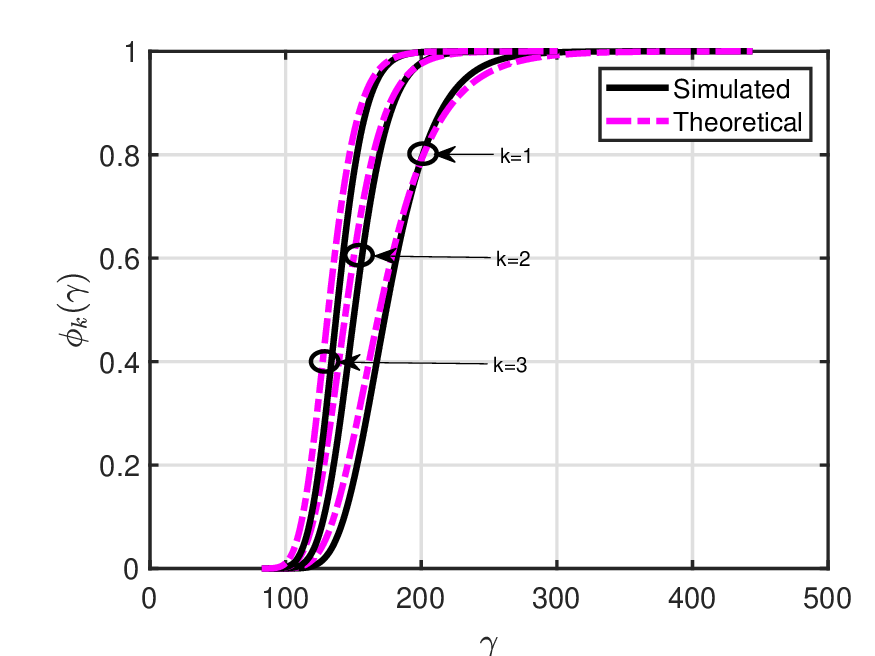}
    \caption{CDF of $\phi_{k}(\gamma)$ for i.n.i.d. RVs with $N_1=12$, $N_2=10$, and $N_3=8$}
    \label{fig:inid_kbl}
\end{figure}
\begin{figure}[h]
    \centering
    \includegraphics[scale=0.5]{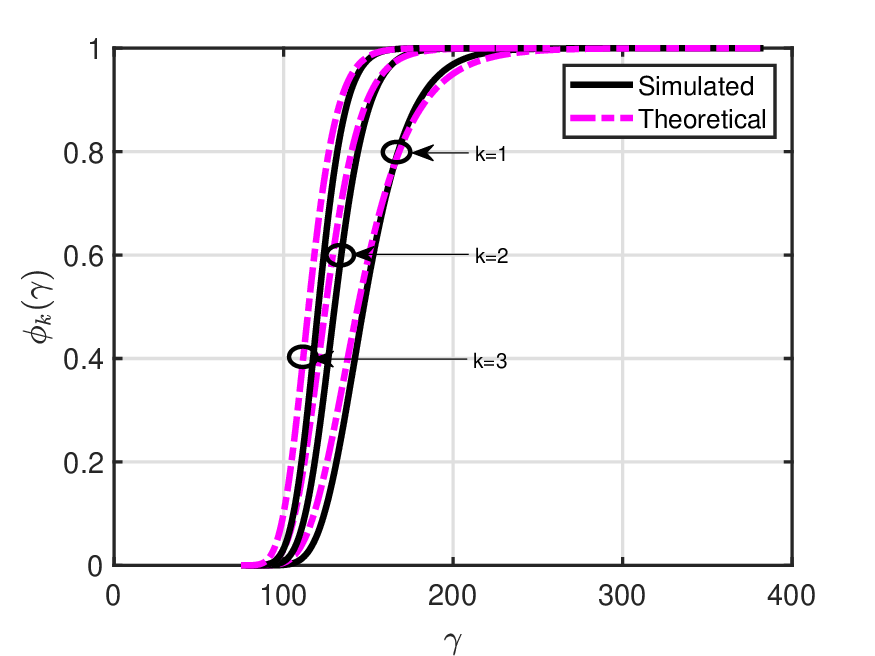}
    \caption{CDF of $\phi_{k}(\gamma)$ for i.i.d. RVs with $N=10$ }
    \label{fig:iid_kbl}
\end{figure}
\subsection{Results for $k$-th maximum}
Here we present the results of the simulation experiments in the case of $k$-th maximum order statistics for both i.n.i.d. and i.i.d. random variables. Here we validate the results of section-III corollary \ref{kthmax}. 
Fig. \ref{fig:inid_kbl} presents the results assuming $N_1=12$, $N_2=10$, and $N_3=8$ where the results are plotted for different values of $k$. Results for i.i.d. case are presented in Fig. \ref{fig:iid_kbl} assuming $N=10$, for different values of $k$. In both cases, we have assumed the number of RVs to be $R=96$, and we can observe that first-order statistics are better than second and third-order statistics.
\subsection{Results for Stochastic ordering}
Fig. \ref{fig:inid_diffgamma} presents the stochastic ordering results. Simulated and theoretical CDFs for maximum order statistics for values of $\gamma_a$=10dB and 30dB, respectively, are plotted considering i.n.i.d. NCCS RVs with one d.o.f. The results in Fig. \ref{fig:inid_diffgamma} are plotted considering the large number of reflecting elements on each RIS ($N_1=60$, $N_2=55$, and $N_3=50$). As the number of reflecting elements increases on each RIS, we can observe that $\lambda_r, \sigma_r^2$ of $\gamma^r$ increases as $\left ( \frac{N_r\pi }{4} \right )^{2},  N_r\left ( 1-\frac{\pi ^{2}}{16} \right )$, respectively as shown in the system model (for $N=60$, $\lambda=2220.7$ and $\sigma=4.79$ ). As the parameter $\lambda$ grows much faster than $\sigma$ for each $\gamma^r$, we can observe that the CDF of $\gamma_{max}$ becomes steeper, as shown in Fig. \ref{fig:inid_diffgamma}. The results in Fig. \ref{fig:inid_diffgamma} also show that even for a large number of reflecting elements, simulated and theoretical CDFs are in good agreement. Further, stochastic ordering has not been characterized before for i.n.i.d. RVs.
\begin{figure}[h]
    \centering
    \includegraphics[scale=0.5]{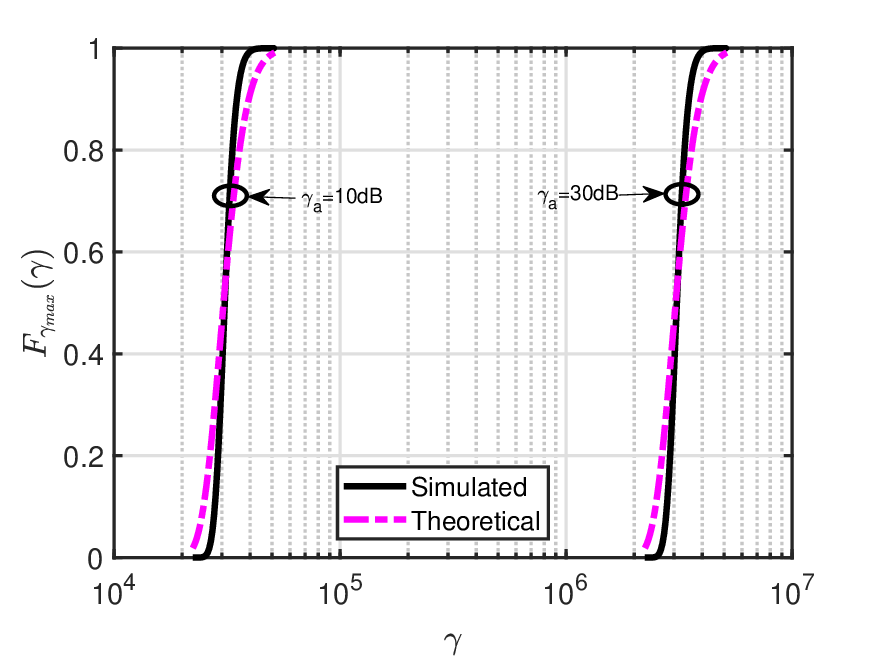}
    \caption{CDF of $\phi_{1}(\gamma)$ for different $\gamma_a$  with $N_1=60$, $N_2=55$, and $N_3=50$ }
    \label{fig:inid_diffgamma}
\end{figure}
\subsection{Results for Outage capacity and Average throughput}
\begin{figure}[h]
    \centering
    \includegraphics[scale=0.5]{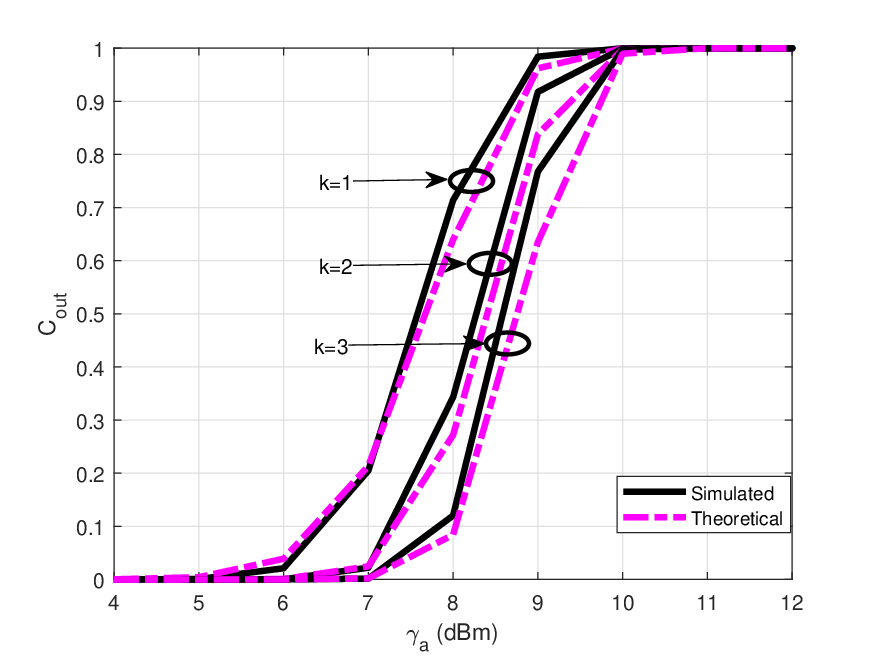}
    \caption{Outage Capacity Vs SNR with $\gamma_{th}$=0 dB}
    \label{fig:outage}
\end{figure}
\begin{figure}[h]
    \centering
    \includegraphics[scale=0.5]{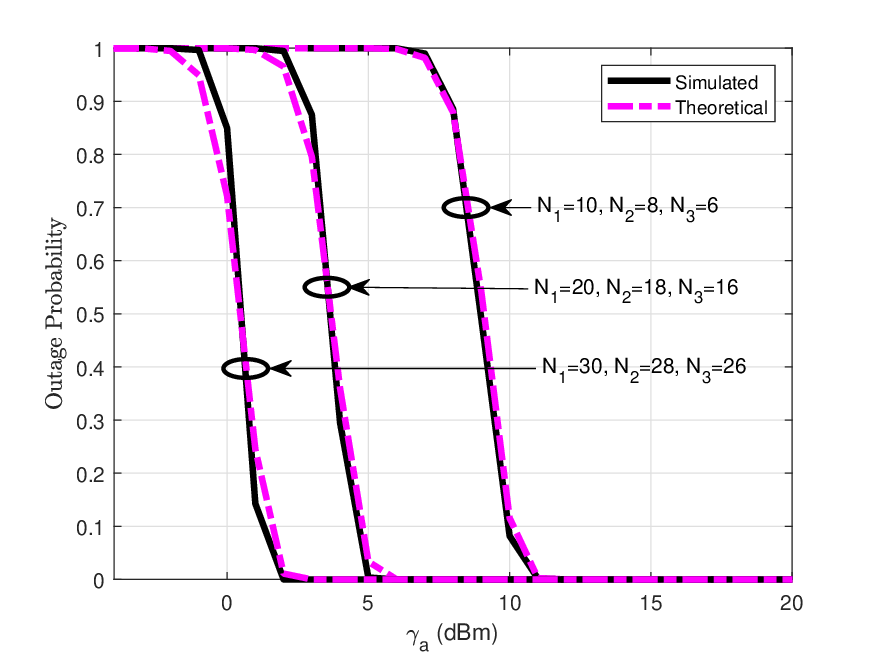}
    \caption{Outage Probability Vs SNR with $\gamma_{th}$=0 dB for different $N$}
    \label{fig:outage_vs_N}
\end{figure}
Here, we present the results for the outage capacity expression derived in (\ref{c_out}) for a multi-RIS communication system. Fig. \ref{fig:outage} presents the results of outage capacity for different values of $k$.
The results are presented considering i.n.i.d. NCCS RVs with $N_1=12$, $N_2=10$, and $N_3=8$ for a threshold of $\gamma_{th}$=0 dB. We can observe that maximum order statistics, i.e., $k=1$, achieve the best performance. Next, Fig. \ref{fig:outage_vs_N} compares the outage probability of a multi-RIS system for different numbers of reflecting elements $(N_1, N_2, N_3)$. We can observe that CDF expression of $k$-th order statistics involves the terms $\tilde{\lambda}$ and $\tilde{\sigma}$, which in turn depends on the number of reflecting elements as $\lambda_r =\left ( \frac{N_r\pi }{4} \right )^{2}$ and $\sigma_r ^{2}=N_r\left ( 1-\frac{\pi ^{2}}{16} \right )$. Hence, it can be clearly seen that the number of reflecting elements plays a crucial role on the system performance. In Fig. \ref{fig:outage_vs_N}, outage probability is plotted for different numbers of reflecting elements. It can be clearly observed that as the number of reflecting elements is increased, outage probability decreases.
\begin{figure}[h]
    \centering
    \includegraphics[scale=0.5]{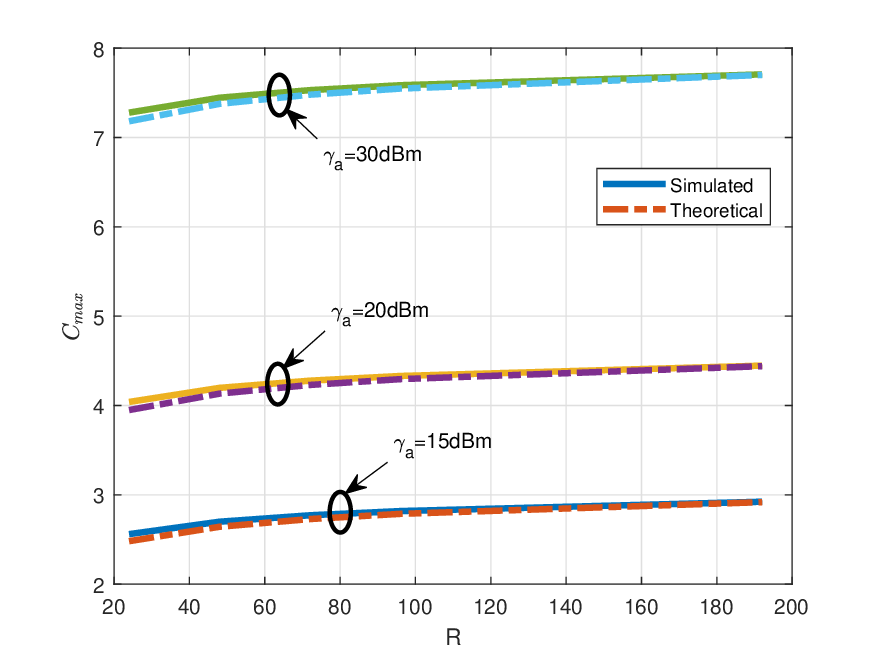}
    \caption{Average throughput Vs R}
    \label{fig:throughput}
\end{figure}
\begin{figure}[h]
    \centering
    \includegraphics[scale=0.5]{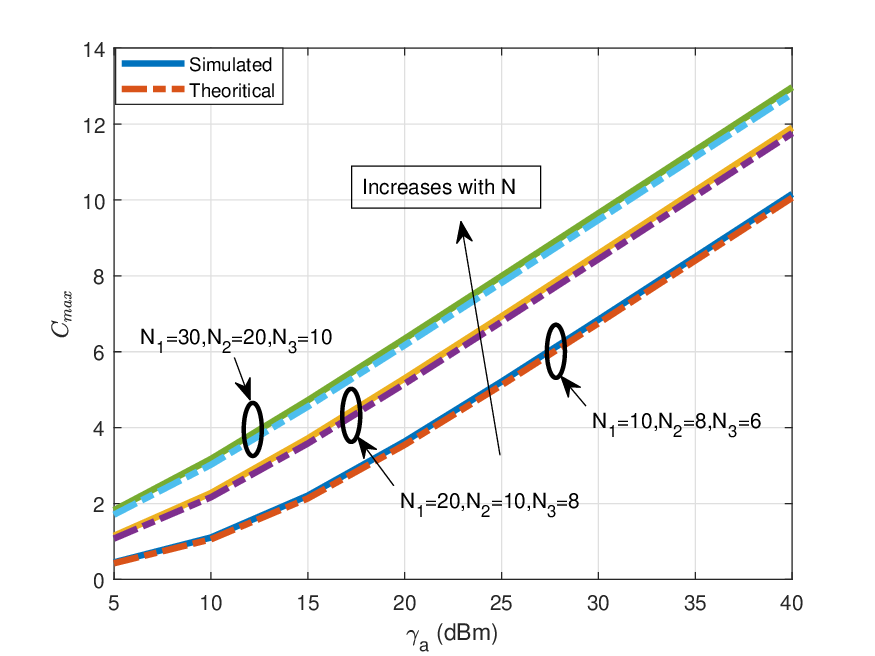}
    \caption{Average throughput of a multi-RIS system for different $N$}
    \label{fig:avg_thr_vs_N}
\end{figure}
\begin{figure}[h]
    \centering
    \includegraphics[scale=0.5]{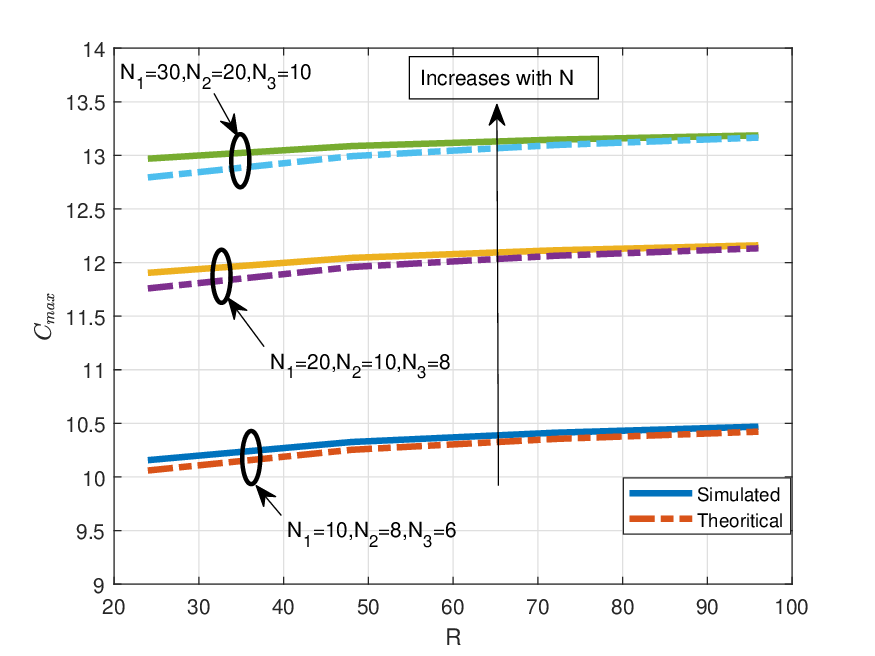}
    \caption{Average throughput of a multi-RIS system for different $N$ and $R$}
    \label{fig:avg_thr_vs_NR}
\end{figure}
Fig. \ref{fig:throughput} shows the average throughput for different values of $\gamma_a$ and R for maximum order statistics. We have evaluated the theoretical average throughput with pdf expression in (\ref{appr}). Here we have used $N_1=12$, $N_2=8$, and $N_3=4$ for simulation experiments.\\
Next, Fig. \ref{fig:avg_thr_vs_N} presents the results of average throughput Vs SNR considering different numbers of reflecting elements. We can also observe from the plots that as the number of reflecting elements increases average throughput of the system is also increasing. We have used $R=24$ for this simulation, and the number of reflecting elements used are $(N_1=10, N_2=8, N_3=6)$, $(N_1=20, N_2=10, N_3=8)$, and $(N_1=30, N_2=20, N_3=10)$.\\
Fig. \ref{fig:avg_thr_vs_NR} compares the average throughput of a multi-RIS system for different $R$ and different numbers of reflecting elements. We know that as the number of RVs (R) increases, the simulated and theoretical average throughputs should converge asymptotically. We can observe that convergence happens even with finite values of $R$. We can observe that the average throughput of the RIS-aided system can be improved with more reflecting elements on each RIS. Fig. \ref{fig:avg_thr_vs_NR} presents the results of average throughput using the reflecting elements as $(N_1=10, N_2=8, N_3=6)$, $(N_1=20, N_2=10, N_3=8)$, and $(N_1=30, N_2=20, N_3=10)$.
\section{{Conclusions }}\label{conclusion}
In this paper, we analyzed the performance of a multi-RIS ($R$ RIS) system where the RISs can have a different number of reflecting elements. Assuming the highest SNR link gets selected for communication, we derived the asymptotic distribution of normalized maximum SNR RV. 
We further derived $k$-th order statistics of i.n.i.d. SNR RVs to deal with scenarios where one is interested in selecting the $k$-th best link. Using our results, we provided outage probability and average throughput expressions for a multi-RIS system.
The simulations showed that the derived asymptotic distribution is in good agreement with the exact distribution, even for moderate values of $R$. 
 
    \begin{appendices}
    \end{appendices}

	\bibliographystyle{IEEEtran}
	\bibliography{reference}

\end{document}